\documentclass{article}

\usepackage{arxiv}

\usepackage[utf8]{inputenc} % allow utf-8 input
\usepackage[T1]{fontenc}    % use 8-bit T1 fonts
\usepackage{hyperref}       % hyperlinks
\usepackage{url}            % simple URL typesetting
\usepackage{booktabs}       % professional-quality tables
\usepackage{diagbox}
\usepackage{physics}
\usepackage{mathtools}
\usepackage{amsthm}
\usepackage{amssymb}
\usepackage{amsmath}
\usepackage{amsfonts}       % blackboard math symbols
\usepackage{nicefrac}       % compact symbols for 1/2, etc.
\usepackage{microtype}      % microtypography
\usepackage{cleveref}       % smart cross-referencing
\usepackage{lipsum}         % Can be removed after putting your text content
\usepackage{graphicx}
\usepackage{natbib}
\usepackage{doi}
\usepackage{makecell}

\usepackage{algorithm}
\usepackage{algorithmicx}
\usepackage{algpseudocode}

\theoremstyle{definition}
\newtheorem{definition}{Definition}
\newtheorem{example}{Example}

\theoremstyle{plain}
\newtheorem{assumption}{Assumption}
\newtheorem{theorem}{Theorem}
\newtheorem{lemma}{Lemma}
\newtheorem{corollary}{Corollary}

\theoremstyle{remark}
\newtheorem*{remark}{Remark}

\newcommand{\cS}{\mathcal{S}}
\newcommand{\I}{\mathcal{I}}
\newcommand{\J}{\mathcal{J}}
\newcommand{\N}{\mathcal{N}}
\newcommand{\bigO}{\mathcal{O}}

\newcommand{\R}{\mathcal{R}}

\newcommand{\dt}{\Delta t}

\newcommand{\exenbr}{t^{l-1}-\epsilon\leq T^{l-1}\leq t^{l-1}}
\newcommand{\tentnbr}{t^{l-1}-\epsilon\leq T^{l-1}_{k^{l-1}}\leq t^{l-1}}

\newcommand{\meandeg}{\langle k\rangle}

\newcommand{\traj}[1]{\overrightarrow{kt}^{#1}}
\newcommand{\TRAJ}[1]{\{K^1,T^1,K^2,T^2,...,K^l,T^l\}}

\newcommand{\cond}{\ensuremath{\;|\;}}
\newcommand{\Rec}{\ensuremath{\text{Rec}}}
\newcommand{\Inf}{\ensuremath{\text{Inf}}}
\newcommand{\Rewire}{\ensuremath{\text{Rewire}}}

\title{Simulation algorithms for Markovian and non-Markovian epidemics}

\author{Guohao Dou\\
	School of Computer and Communication Sciences\\
	EPFL\\
	1015 Lausanne \\
	\texttt{guohao.dou@epfl.ch} \\
}

\renewcommand{\shorttitle}{Simulation algorithms for Markovian and non-Markovian epidemics}

\hypersetup{
pdftitle={Simulation algorithms for Markovian and non-Markovian epidemics},
pdfsubject={},
pdfauthor={Guohao Dou},
pdfkeywords={},
}

\begin{document}
\maketitle

\begin{abstract}
	Researchers have employed stochastic simulations to determine the validity of their theoretical findings and to study analytically intractable spreading dynamics. In both cases, the correctness and efficiency of the simulation algorithm are of paramount importance. We prove in this article that the Next Reaction Method and the non-Markovian Gillespie algorithm, two algorithms for simulating non-Markovian epidemics, are statistically equivalent. We also study the performance and applicability under various circumstances through complexity analyses and numerical experiments. In our numerical simulations, we apply the Next Reaction Method and the Gillespie algorithm to epidemic simulations on time-varying networks and epidemic simulations with cooperative infections. Both tasks have only been done using the Gillespie algorithm, while we show that the Next Reaction Method is a good alternative. We believe this article may also serve as a guide for choosing simulation algorithms that are both correct and efficient for researchers from epidemiology and beyond. 
\end{abstract}

% keywords can be removed
\keywords{Gillespie Algorithm \and Epidemic Simulation}

\section{Introduction}
Simulations have become a highly relevant part of high-stakes policymaking, especially in the context of disease control and prevention. In 2020, a team from Imperial College published a report\cite{ferguson2020impact} simulating the spread of COVID-19 with different levels of government intervention. The report predicted dire consequences if no government action were taken, convincing the UK government to issue a series of new restrictions. 

We believe that our understanding of epidemic simulations should keep up with their ever-growing political relevance and social impact. As pointed out in \cite{kiss2017mathematics}, from classical models assuming well-mixed populations (e.g. \cite{kermack1927contribution}) to sophisticated agent-based models (e.g. \cite{ferguson2020impact}), we see an increase in both model realism and model complexity, where the former is desirable while the latter is not. Like the authors of \cite{kiss2017mathematics}, we also believe that epidemic simulations on networks sit in the middle of the complexity spectrum, allowing us to maintain the balance between mathematical rigor and the ease of incorporating real-world data. 

% challenges in epi sim on networks 
A few technical challenges are encountered while making epidemic simulations on networks more realistic. Firstly, researchers tend to assume Poissonian transitions between compartments, while epidemiologists have found that recovery rates can increase over time, as opposed to being constant, as suggested by exponentially distributed recovery times \cite{Lloyd2001}. As for the ongoing COVID-19 pandemic, the authors of \cite{lu2020epidemiological} find that the best-fitting model for the incubation period based on all 1158 patients is the Weibull distribution, with the Gamma distribution being a close second. Secondly, networks on which epidemic processes take place are usually assumed to be static, which is reasonable when the network evolves on a time scale much slower than the epidemic, but is clearly unjustified for faster network evolutions such as commuting. This assumption also under-utilizes data from digital tracking\cite{salathe2012digital}\cite{ferretti2020quantifying}.

% inspirations from the algs of chem
In search of a disciplined way of introducing realism to epidemic simulations, we survey computational chemistry for inspiration. In an environment where different types of molecules coexist, a set of coupled chemical reactions occurs at different rates and the times to their occurrence follow exponential distributions. For each iteration, both the reaction type and the time of reaction are identified. The chosen reaction will then be executed and the molecule count will be updated. 

In his 1976 paper\cite{gillespie1976general}, Gillespie proposes the ``First Reaction Method'' which, in each iteration, generates a tentative time for each chemical reaction and then locates the most immediate reaction to occur. The downside of the First Reaction Method is that it needs to generate as many random numbers per iteration as there are reactions. To fix this inefficiency, Gibson et al. propose Next Reaction Method~(NRM)\cite{Gibson2000}, which maintains a priority queue of tentative times and ensures that random numbers are rescaled and reused whenever possible, reducing the number of random numbers generated to one per iteration. NRM can be readily extended to the non-Markovian setting, since even though rescaling no longer works, we can still reuse generated tentative times in the priority queue. 

In the same 1976 paper, Gillespie also proposes the ``Direct Method'' which generates the reaction time and reaction type directly without sampling tentative times. This method makes extensive use of the property that the minimum of exponentially distributed random variables also follows an exponential distribution, whose rate is the sum of all individual rates. 

Considerable efforts have gone into extending the Direct Method to more general settings. Boguñá et al.\cite{Boguna2014} propose the non-Markovian Gillespie Algorithm~(nMGA) which works with non-Poissonian transitions, allowing tentative times on individual ``reactions'' to follow arbitrary distributions. In 2015, Vestergaard et al.\cite{Vestergaard2015} extend the method to simulating epidemic processes on time-varying networks, by assuming that network changes only happen at given, fixed time intervals. The authors also venture into the general non-Markovian case, even though they end up resorting to approximation. For the rest of this article, we refer to methods akin to the Direct Method as ``Gillespie algorithms'' to stay in line with the nomenclature in the literature, noting their defining feature that no realization of tentative times needs to be generated to compute their minimum. 

With all interevent times assumed to be exponentially distributed, the equivalence between the First Reaction Method and the Direct Method is proved in \cite{gillespie1976general} and the equivalence between the Direct Method and the Next Reaction Method is proved in \cite{Gibson2000}. In the general non-Markovian case, we show in this article that, just as NRM and Direct Method are equivalent in the Markovian case, their non-Markovian extensions are also equivalent. 

\section{Event Emitters}
In this section, we define event emitters and give examples of event emitters for both an SIS and an SIR epidemic on an undirected simple graph.
\begin{definition}{(Event emitters)}\label{def:ee}
  Let $\I$ be the space of identifiers and $\mathcal{P}(\I)$ its power set. $\cS$ is the state space of the system. An event emitter $e$ is defined by the tuple $(\mbox{ID}, \Psi, f, C, R)$ where
  \begin{itemize}
      \item $\mbox{ID}\in\I$ is the ID of the event emitter $e$.
      \item $\Psi$ is the cumulative distribution function (CDF) specifying the interevent distribution, with $\Psi(t)=0, \forall\; t\leq 0$.
      \item $f:\cS\to\cS$ is the state transition function.
      \item $C: \cS\to\mathcal{P}(\I)$ maps a state to the set of IDs of event emitters that will be created upon execution.
      \item $R: \cS\to\mathcal{P}(\I)$ maps a state to the set of IDs that will be removed upon execution, with the additional requirement that $\forall s\in\cS, \mbox{ID}\in R(s)$, i.e., event emitters are self-removing.
  \end{itemize}
\end{definition}
$\Psi$ specifies when an event is going to be triggered, and $f, C$ and $R$ specify what is going to happen if this event emitter is chosen for execution. $C$ and $R$ are functions of the state instead of subsets because for non-trivial applications, we need the state of the system to compute what to create or remove. Also, we require all event emitters to delete themselves upon execution by ensuring $\forall s\in\cS, \mbox{ID}\in R(s)$. 

During implementation, some integrity checking is recommended. For example, if $E$ is the set of event emitters in the system, we may want to ensure that $\forall s\in\cS, R(s)\subseteq E$ and $C(s)\cap E=\emptyset$, i.e., we disallow the removal of nonexistent event emitters and the creation of duplicates. 

This way of orchestrating simulation tasks improves code modularity by encapsulating the specifics inside well-defined event emitters. Also, event emitters are designed to be ``immutable'', that is, their configurations are not meant to be directly modified and any modification must be achieved through removal and subsequent renewal.

\begin{example}
  We define event emitters for an SIS epidemic on a network. Inf($x\to y$) stands for a pending infection from node $x$ to node $y$ along the edge $(x,y)$, while Rec($x$) stands for a pending recovery of node $x$. The interevent distribution $\Psi$ is intentionally omitted in Table~\ref{tableSIS}, since it is usually specified by the particular natural history of the epidemic. We do the same for an SIR epidemic on a network, as shown in Table~\ref{tableSIR}. Note that in an SIR model, once a node recovers, it is no longer threatened by its infected neighbors. The set of event emitters newly created by Rec($x$) is thus $\emptyset$. This design makes it clear that we assume all per-edge infections and all recoveries to be mutually independent. 
\end{example}

\begin{table}
  \centering
  \caption{\textbf{Event emitters for an SIS epidemic on a network.}}
  \begin{tabular}{|c|c|c|}
  \toprule
             & Infection                          & Recovery                      \\ \hline
  ID         & Inf($x\to y$)                      & Rec($x$)                      \\ \hline
  f          & $y$.state $\gets$ I (Infected)         & $x$.state $\gets$ S (Susceptible) \\ \hline
  C          & \makecell{$\{\text{Rec}(y)\}$\\$\cup$\\$\{\text{Inf}(y\to z) \;|\; z\in \N(y), z.\text{state=S}\}$} & \makecell{$\{\text{Inf}(z\to x)\;|\; z\in \N(x), z\text{.state=I}\}$} \\ \hline
  R          & \makecell{$\{\text{Inf}(z\to y)\;|\; z\in \N(y), z\text{.state=I}\}$} & \makecell{$\{\text{Rec}(x)\}$ \\ $\cup$ \\$\{\text{Inf}(x\to z)\;|\; z\in \N(x),z\text{.state=S}\}$} \\ 
  \bottomrule
  \end{tabular}
  \begin{flushleft} 
    Table notes: $\{\text{Inf}(y\to z) \;|\; z\in \N(y), z.\text{state=S}\}$ means ``Inf($y\to z$) for all susceptible neighbors of $y$''; we also use a state reassignment instead of a state transition function for brevity.
  \end{flushleft}
  \label{tableSIS}
\end{table}

\begin{table}
  \centering
  \caption{\textbf{Event emitters for an SIR epidemic on a network.}}
  \begin{tabular}{|c|c|c|}
  \toprule
             & Infection                          & Recovery                      \\ \hline
  ID         & Inf($x\to y$)                      & Rec($x$)                      \\ \hline
  f          & $y$.state $\gets$ I         & $x$.state $\gets$ R (Recovered) \\ \hline
  C          & \makecell{$\{\text{Rec}(y)\}$\\$\cup$\\$\{\text{Inf}(y\to z) \;|\; z\in \N(y), z.\text{state=S}\}$} & \makecell{$\emptyset$} \\ \hline
  R          & \makecell{$\{\text{Inf}(z\to y)\;|\; z\in \N(y), z\text{.state=I}\}$} & \makecell{$\{\text{Rec}(x)\}$ \\ $\cup$ \\$\{\text{Inf}(x\to z)\;|\; z\in \N(x),z\text{.state=S}\}$} \\ 
  \bottomrule
  \end{tabular}
  \label{tableSIR}
\end{table}

\section{Next Reaction Method}
Gibson et al. propose the Next Reaction Method (NRM)\cite{Gibson2000} as an improvement upon the First Reaction Method proposed in \cite{gillespie1976general}. The core idea of the improvement is that previously generated interevent times are reused. In the First Reaction Method, a tentative time is generated for each event emitter after the execution of any event in the system. Gibson et al. show that the vast majority of tentative times can be reused via careful rescaling without sacrificing statistical exactness, when interevent times follow exponential distributions and the set of event emitters is fixed. 

Adapting NRM to epidemic simulations on networks is not completely straightforward. First and foremost, while coupled chemical reactions have an immutable set of reactants and reactions with mutable Poisson rates, simulation tasks in epidemiology have a mutable set of pending events with pre-defined, immutable interevent distributions. The mutability of pending events calls for more bookkeeping in algorithm design, while the immutability of interevent distributions offers optimization opportunities. Secondly, it is not clear whether the highly efficient NRM can be applied to the more general non-Markovian case. The original NRM paper\cite{Gibson2000} states that ``In general, Monte Carlo simulations assume statistically independent random numbers, so it is usually not legitimate to re-use random numbers. In this particular special case, we shall prove that it is legitimate.'' However, the proof in \cite{Gibson2000} is done assuming exponential distributions. 

Important clarification of terms needs to be made: what we refer as non-Markovian epidemics are conceptually closer to what the authors of \cite{Gibson2000} refer to as ``time-varying Markov processes'' in section 4 of \cite{Gibson2000}, with a focus on non-exponential interevent distributions. The authors of \cite{Gibson2000} prove the equivalence of NRM to Gillespie's Direct Method when all interevent times are assumed to be exponential. In this article, we prove the equivalence between NRM and nMGA\cite{Boguna2014} while making no assumption on interevent distributions. 

In this section, we first reformulate NRM using the notation of event emitters. We then prove some useful properties of NRM that will aid in the proof of the statistical equivalence between NRM and nMGA later. Lastly, we discuss some implementation details and analyze their complexity. 

First and foremost, we define some quantities needed to describe the algorithm. 
\begin{definition}(Tentative time of execution) We denote the tentative time of execution of emitter $i$ in iteration $l$ as $T^l_i$. It is the time when $e^l_i$ will execute if no other event emitter removes it first. 
\end{definition}

\begin{definition}(Time and location of the $l$-th execution)
  $T^l$ is the time of the $l$-th execution and $K^l$ is the event emitter that triggers the $l$-th execution. Only the event emitter with the most immediate event may execute, meaning that 
  \[K^l \coloneqq \mathrm{argmin}_i T^l_i,\quad T^l \coloneqq \min_i T^l_i.\]
\end{definition}

For the rest of this article, we adopt the indexing convention where superscripts represent iteration indices, and subscripts represent emitter IDs. For example, $T^l$ reads ``the time of the $l$-th execution'' and $\Psi_i$ reads ``the interevent distribution of event emitter $i$.'' 

We describe NRM in Algorithm~\ref{alg:nrm}.

\begin{algorithm}
  \caption{NRM}\label{alg:nrm}
  \begin{enumerate}
    \item (Initialization Phase)
    \begin{itemize}
        \item $E^1=\{e^1_1 ... e^1_{n^1}\}$ is the initial set of event emitter IDs
        \item The system is at initial state $s\in\cS$
        \item Set the clock $T^0 \gets 0$
        \item Set iteration index $l\gets 1$
        \item Generate $T_i^1\sim\Psi_i,\; i=1,2,...,n^1$
    \end{itemize}
    \item (Main Phase)
    \begin{enumerate}
        \item The algorithm terminates if $|E^l| = 0$
        \item Find argmin by $K^l\gets \text{argmin}_i T^l_i$
        \item Update the clock by $T^l\gets \min_i T^l_i$
        \item Update emitters by $E^{l+1}\gets (E^l\backslash R_{K^l}(s))\cup C_{K^l}(s)$
        \item For each $j\in C_{K^l}(s)$, generate interevent time by $\tilde{T}^{l+1}_{j}\sim \Psi_j$, and then convert it to time of execution by $T^{l+1}_j \gets \tilde{T}^{l+1}_j + T^l$
        \item Apply state transition by $s\gets f_{K^l}(s)$
        \item Increment iteration index by $l\gets l+1$
        \item Go back to step (a)
    \end{enumerate}
  \end{enumerate}
\end{algorithm}

We extract the following update rule from Algorithm~\ref{alg:nrm}.
\begin{align*}
  T^{l+1}_i = \begin{cases}
    \tilde{T}^{l+1}_i + T^l  \;\;\; &\text{ if } i\in C_{K^l} (\text{``freshly generated''})\\
    T^l_i  \;\;\; &\text{ if } i\notin C_{K^l} (\text{``reused''})
  \end{cases}
\end{align*}
where $\tilde{T}^{l+1}_i\sim\Psi_i$. The notion of ``freshly generated'' and ``reused'' is self-explanatory, since only event emitters in $C_{K^l}$ require random number generation (RNG). In fact, for each iteration, tentative times are either freshly generated or reused, forming a partition. For the rest of the article, we denote the set of emitter IDs with reused tentative times in iteration $l$ as $\R^l$. By definition, $\R^l\coloneqq E^l\setminus C_{K^{l-1}}$.

Having recognized that the outcome of an algorithm is completely determined by the min and argmin generated in each iteration, we define the trajectory up to iteration $l-1$ as a sequence of realizations.
\begin{definition}(trajectory up to $l-1$)\label{def:traj}
  The trajectory of an algorithm up to iteration $l-1$ is denoted as
  \begin{align*}
    \{\traj{l-1}\} \coloneqq \{K^1=k^1, T^1=t^1,...,K^{l-1}=k^{l-1}, T^{l-1}=t^{l-1}\}.
  \end{align*}
The sequence $\{t^0, t^1, ..., t^l, ...\}$ is a point process on the real half-line, and only at these points can event emitters be created or removed. 
\end{definition}

We now look into the distribution of $T^l_i$. We first note that the distribution of $T^l_i$ is always understood as being conditioned on the trajectory up to $l-1$. When we write down the distribution of $T^l_i$, what we really mean is
\begin{align*}
  \Pr\{T^l_i\leq t\cond \traj{l-1}\}.
\end{align*}
Once a tentative time is generated, it is never modified except for the possibility of getting removed by the execution of some other event emitters. However, since only the minimum and argmin are recorded in the trajectory, we can view $\Psi_i$ as prior distributions and the process of locating the minimum as updating priors in order to obtain posterior distributions. When we speak of the distribution of $T^l_i$, we refer to the posterior distribution conditioned on the trajectory up to $l-1$. 

Knowing $\traj{l-1}$, we can recover the set of event emitters $E^l$ and its reused partition $\R^l$.

\begin{example}
  To make Algorithm~\ref{alg:nrm} less abstract, we go through the first three iterations of a minimal example illustrated in Fig.~\ref{figNRMExample}. The behavior of event emitters is defined in Table~\ref{tableSIS}. For clarity, we use integers instead of full emitter IDs for subscripts. 

  We observe that some tentative times are freshly generated, such as $T^1_1, T^1_2, T^1_3, T^2_3, $ $T^2_4, T^3_5$, while others are reused, such as $T^2_1, T^2_2, T^3_3, T^3_4$. Therefore, $\R^1=\emptyset$, $\R^2=\{1, 2\}$, $\R^3=\{3, 4\}$.
\end{example}

\begin{figure}[!h]
  \centering
  \includegraphics[scale=0.13]{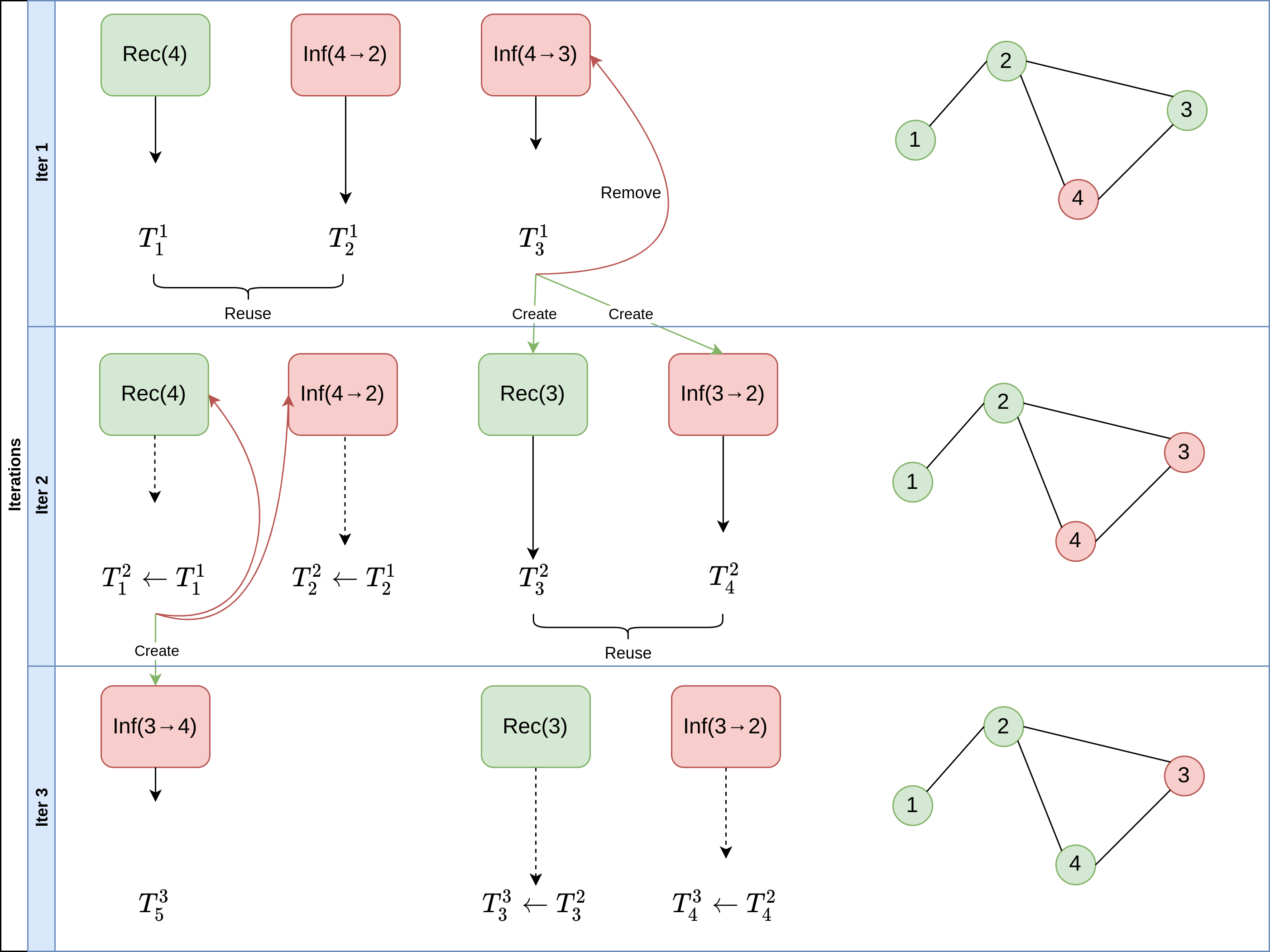}
  \caption{{\bf NRM Example: an SIS epidemic on a small network.} The left side of the figure depicts the evolution of event emitters, while the right side keeps track of the state of the network. Each box represents an event emitter, with a downward-pointing arrow whose length is the tentative time of execution. As we see in the first iteration, Inf(4$\to$3) manages to draw the shortest straw, leading to the infection of node 3, the removal of Inf(4$\to$3), as well as the creation of Rec(3) and Inf(3$\to$2). As for the second iteration, the first two dashed arrows reflect the fact that no additional RNG is required for these two as the numbers from the previous iteration are reused. Meanwhile, the two newly created event emitters, Rec(3) and Inf(3$\to$2), do require RNG, as indicated by the solid arrows.}
  \label{figNRMExample}
\end{figure}

Having described NRM, we proceed to demonstrate some of its properties that will prove crucial in establishing the statistical equivalence between NRM and nMGA. We start by making an assumption on the independence of freshly generated tentative times. 
\begin{assumption}\label{assumption:assumption}
  Let $\{T^l_i\}_{i\in \mathcal{J}} \subseteq \{T^l_i\}_{i\in E^l}$ be an arbitrary subset of tentative times in iteration $l$ and $T^l_j, j\in\mathcal{J}$ a freshly generated tentative time. The joint cumulative distribution function can be factorized as follows,
  \begin{align*}
    \Pr\{\forall i\in \mathcal{J}, T^l_i\leq t_i \cond\traj{l-1} \} = \Pr\{T^l_j\leq t_j \cond\traj{l-1}\}\cdot \Pr\{\forall i\in \mathcal{J}\backslash\{j\}, T^l_i\leq t_i\cond\traj{l-1}\}.
  \end{align*}
\end{assumption}
 
\begin{remark}
  Note that Assumption~\ref{assumption:assumption} is weaker than mutual independence. For example, in Fig.~\ref{figNRMExample}, even though we can safely do
  \[\Pr\{T^2_3\leq t_3, T^2_4\leq t_4\cond\traj{1}\} = \Pr\{T^2_3\leq t_3\cond\traj{1}\} \Pr\{T^2_4\leq t_4\cond\traj{1}\},\]
  we cannot factorize $\Pr\{T^2_1\leq t_1, T^2_2\leq t_2 \cond\traj{1}\}$ based on Assumption~\ref{assumption:assumption} alone since both $T^2_1\gets T^1_1$ and $T^2_2\gets T^1_2$ are reused and \textbf{NOT} freshly generated. This difficulty leads us to examine the effect of reusing tentative times and check if the algorithm introduces any dependency between previously independent random variables. 
\end{remark}

For the rest of the article, we assume the tentative times $T^l_i$ to be continuous random variables. We also adopt the following vectorized shorthands
\begin{itemize}
  \item for an ID set $\J$ and a scalar $t$
  \begin{align*}
    \Pr\{\vec{T}^l_\J\leq t\cond\traj{l-1}\}\coloneqq\Pr\{\forall j\in\J, T^l_j\leq t\cond\traj{l-1}\}.
  \end{align*}
  \item for an ID set $\J$ and a vector $\vec{t}_\J$
  \begin{align*}
    \Pr\{\vec{T}^l_\J\leq \vec{t}_\J\cond\traj{l-1}\}\coloneqq\Pr\{\forall j\in\J, T^l_j\leq t_j\cond\traj{l-1}\}.
  \end{align*}
\end{itemize}
It should be understood that in $\Pr\{\vec{T}^l_\J\leq t\cond\traj{l-1}\}$ or $\Pr\{\vec{T}^l_\J\leq \vec{t}_\J\cond\traj{l-1}\}$, we only consider the case where $t > t^{l-1}$ or $\forall j\in\J, t_j > t^{l-1}$. Otherwise, both would be zero. 

We are now ready to prove Theorem~\ref{thm:indep}. The mutual independence of tentative times has been proved by \cite{Gibson2000} for the case in which all interevent times follow exponential distributions. Here we prove the mutual independence of non-exponential tentative times, with the additional understanding that the execution of events can create/remove other event emitters.
\begin{theorem}\label{thm:indep}
  For any $l>0$, $\{T^l_i\}_{i\in E^l}$ consists of mutually independent random variables given $\traj{l-1}$. 
\end{theorem}
The proof is given in \nameref{S2_Appendix}. 

\begin{corollary}(Corollary of Theorem~\ref{thm:indep})\label{corollary:inductOne}
  If $\{T^l_i\}_{i\in\J}$ consists of reused tentative times, then their joint CDF can be expressed as 
  \begin{align*}
    \Pr\{\vec{T}^l_\J\leq \vec{t}_\J\cond\traj{l-1}\} = \prod_{i\in\J} \Pr\{T^{l-1}_i\leq t_i\cond T^{l-1}_i > t^{l-1}, \traj{l-2}\}.
  \end{align*}
\end{corollary}

We have already managed to express the distribution of $T^l_i$ in terms of the distribution of $T^{l-1}_i$ from the previous iteration. To further our analysis, we now express the distribution of $T^l_i$ in terms of the user-specified interevent distributions introduced in Defn.~\ref{def:ee}. To this end, we first define the time of creation of an event emitter. 
\begin{definition}(time of creation)
  $t^l_i$ is the time when event emitter $i$ of the $l$-iteration was created.
\end{definition}
From a algorithmic point of view, keeping track of times of creation is trivial: simply store the clock value in the event emitter when it is created. However, mathematically speaking, we can recover times of creation only if we have complete information of the previous trajectory $\{\traj{l-1}\}$. For initially present event emitters, their times of creation are unsurprisingly $t^0=0$. For event emitters created in iteration $l$, their times of creation are $t^l$. In both cases, $t^l_i$ only takes values from $\{t^0, t^1, t^2,..., t^{l-1}\}$. It is also possible for event emitters with the same ID to be recreated after its removal, in which case the time of creation points to the most recent creation. 

\begin{lemma}\label{lemma:TDist}
  Let $\Psi_i$ be the CDF of the interevent distribution of emitter $i$, and $t^l_i$ the time of creation of the event emitter $e^l_i$. The CDF of $T^l_i$ is
  \begin{align}\label{eqn:TDistLong}
    \Pr\{T^l_i\leq t\cond\traj{l-1}\} = \frac{\Psi_i(t-t^l_i) - \Psi_i(t^{l-1}-t^l_i)}{1 - \Psi_i(t^{l-1} - t^l_i)}.
  \end{align}
\end{lemma}
The proof is given in \nameref{S3_Appendix}.

To make Eq~(\ref{eqn:TDistLong}) less laborious, we adopt the notations in\cite{Boguna2014} as a shorthand. Let $\Psi$ be the user-specified interevent distribution such as exponential, Gamma or Weibull distributions. The distribution shifted by $t$ is expressed as in Eq~(\ref{eqn:conddf}),
\begin{align}
  \Psi(\tau\cond t) \coloneqq \frac{\Psi(\tau + t) - \Psi(t)}{1 - \Psi(t)}, \;\;\;\;\psi(\tau\cond t) \coloneqq \pdv{\Psi(\tau\cond t)}{\tau} = \frac{\psi(\tau + t)}{1 - \Psi(t)}.\label{eqn:conddf}
\end{align}
A simple pattern matching using Eq~(\ref{eqn:TDistLong}) and Eq~(\ref{eqn:conddf}) yields Eq~(\ref{eqn:tdist})
\begin{align}
\Pr\{T^l_i \leq t\cond\traj{l-1}\} =  \Psi_i(t-t^{l-1}\cond t^{l-1} - t^l_i). \label{eqn:tdist}
\end{align} 

\begin{example}
  We take another look at the example in Fig.~\ref{figNRMExample}. We recall that these tentative times are well-defined only when the previous trajectory is given. In this particular case, the fact that $T^2_1$ exists indicates that $T^1_1$ must be larger than the minimum of tentative times in the first iteration, $t^1$, giving us Eq~(\ref{eqn:examplenrm}):
  \begin{align}
    \Pr\{T^2_1\leq t\cond\traj{1}\} = \Pr\{T^1_1\leq t\cond T^1_1>t^1\} = \Psi_1(t - t^1\cond t^1). \label{eqn:examplenrm}
  \end{align}
  Indeed, Eq~(\ref{eqn:examplenrm}) is the special case of Eq~(\ref{eqn:tdist}) where $l=2$ and $t^2_1=t^0=0$.
\end{example}

We now look into the per-iteration computational costs and important implementation details of NRM. We make the following assumptions:
\begin{itemize}
  \item The number of active event emitters is obviously always upper-bounded by $N$. In particular, for epidemics on networks, $N$ scales as $\bigO(V+E)$.
  \item The number of event emitters created/removed per iteration is $\bigO(1)$. This assumption makes sense for epidemic simulations on networks since these numbers are usually upper-bounded by the mean degree of the network (e.g. $\Inf(x\to y)$ for each susceptible neighbor $y$ of $x$).
  \item The complexity of generating each interevent time is $\bigO(1)$; generated interevent times follow the distributions indicated in Algorithm~\ref{alg:nrm} exactly and no approximation is involved.
\end{itemize}
These assumptions hold for the complexity analyses of all the algorithms in this article. 

Since NRM finds the minimum tentative time proposed by existing event emitters, we make use of a min-heap. The min-heap data structure should support the following three operations:
\begin{itemize}
  \item \textbf{RemoveMin}: Retrieves and removes the top of the heap, which is the minimum tentative time; time complexity is $\bigO(\log N)$.
  \item \textbf{Insert}: Inserts an event emitter and its associated tentative time into the heap; time complexity is $\bigO(\log N)$.
  \item \textbf{RemoveByID}: Removes an event emitter in the heap (not necessarily at the top) by its ID, while maintaining the heap property; time complexity is $\bigO(\log N)$.
\end{itemize}
While \textbf{RemoveMin} and \textbf{Insert} are standard, operation \textbf{RemoveByID} is not, since one can only access the element at the top in a classic binary heap implementation. To achieve this, we need a hash table that maps event emitter IDs to their locations in the heap and perform heap percolate-down from the location indicated by the hash table, instead of the top. The time complexity is still $\bigO(\log N)$ since it cannot incur more operations than a percolate-down from the top.

\textbf{RemoveMin} is invoked once per iteration, which is $\bigO(\log N)$. Because we assume that the number of event emitters created/removed is always $\bigO(1)$, the per-iteration time complexity of \textbf{Insert} and \textbf{RemoveByID} is also $\bigO(\log N)$. 

The per-iteration time complexity of Algorithm~\ref{alg:nrm} is thus $\bigO(\log N)$.

\section{Gillespie Algorithms}
The Gillespie algorithm, also termed the ``Direct Method'' in \cite{gillespie1976general}, aims to reduce random number generation. In Algorithm~\ref{alg:nrm}, we need to generate as many random numbers as the number of newly created event emitters in each iteration. In all variants of the Gillespie algorithm we are about to introduce in this section, however, exactly two random numbers need to be generated per iteration. We call these algorithms ``Gillespie algorithms'' and introduce them from the versatile but slow ones to the limited but fast ones. 

\subsection{General Non-Markovian Epidemics}
Boguñá et al introduce in \cite{Boguna2014} the non-Markovian Gillespie Algorithm~(nMGA) and extend the Markovian Gillespie algorithm in \cite{gillespie1976general} to the more general non-Markovian scenario. In addition to the reduced random number generation mentioned before, nMGA also lends itself better to the case in which computing event rates is straightforward while generating interevent times requires nontrivial numerical methods, as illustrated by the example of cooperative infections given in \cite{Boguna2014}. The authors of \cite{Boguna2014} also propose an approximate version of nMGA, which approximates the actual event rate function with a staircase function. In this article, to prevent confusion, we call the exact version of nMGA ``nMGA-Exact'' and the approximate version ``nMGA-Approx''.

We first review both versions of nMGA, the exact~(Algorithm~\ref{alg:nMGAExact}) and the approximate~(Algorithm~\ref{alg:nMGAApprox}), using the same notations as Algorithm~\ref{alg:nrm}. 

\begin{algorithm}
  \caption{nMGA-Exact}\label{alg:nMGAExact}
  \begin{enumerate}
    \item (Initialization Phase)
    \begin{itemize}
      \item $E^1$ is given as the initial set of event emitters
      \item The system is at initial state $s\in\cS$
      \item Set the clock $T^0\gets 0$
      \item Set iteration index $l\gets 1$
      \item Store times of creation $\forall i\in E^1, t^1_i\gets 0$
    \end{itemize}
    \item (Main Phase)
    \begin{enumerate}
      \item Algorithm terminates if $|E^l| = 0$
      \item Compute the global survival probability $\Phi^l(t) = \prod_{i\in E^l} \left[1-\Psi_i(t\cond t^{l-1}-t^l_i)\right]$
      \item Generate $\Delta^l\sim\Phi^l$ and progress in time by $T^l\gets T^{l-1}+\Delta^l$
      \item Compute PMF with weights $\Pi^l(i)\coloneqq \frac{\psi_i(0\cond t^l-t^l_i)}{\sum_i\psi_i(0\cond t^l-t^l_i)}$ and generate $K^l\sim\Pi^l$
      \item Update emitters by $E^{l+1}\gets (E^l\backslash R_{K^l}(s))\cup C_{K^l}(s)$
      \item Store times of creation $\forall i\in C_{K^l}(s), t^l_i\gets T^l$
      \item Apply state transition by $s\gets f_{K^l}(s)$
      \item Increment $l\gets l+1$
      \item Go back to step (a)
    \end{enumerate}
  \end{enumerate}
\end{algorithm}

\begin{algorithm}
  \caption{nMGA-Approx}\label{alg:nMGAApprox}
  We replace step (b) and (c) in Algorithm~\ref{alg:nMGAExact} with the following two steps. First compute 
  \[\lambda^l_i \coloneqq \psi_i(0\cond t^{l-1} - t^l_i).\]
  Then generate
  \[\Delta^l\sim \text{Exp}\left(\sum_i \lambda^l_i\right),\]
  and progress in time as usual by $T^l \gets t^{l-1} + \Delta^l$. Everything else stays the same as Algorithm~\ref{alg:nMGAExact}.
\end{algorithm}

We then prove that Algorithm~\ref{alg:nrm} and Algorithm~\ref{alg:nMGAExact} are statistically equivalent. Each of Algorithm~\ref{alg:nrm} and Algorithm~\ref{alg:nMGAExact}, after $l$ iterations of execution, produces a realization of a $2l$-dimensional random vector $\TRAJ{l}$. We thus say that two algorithms are equivalent if for any $l>0$, their trajectories up to $l$ are equal in distribution. 

Since the random vector $\TRAJ{l}$ contains both discrete and continuous random variables, standard notations such as probability density function for random vectors cannot be applied directly. We can, however, work around this difficulty by conditioning. Again, let $\traj{l-1}$ be the previous trajectory. Note that the density function of $T^l$
\begin{align*}
  f_{T^l}(t^l) \coloneqq \dv{}{t}\Bigr|_{t^l} \Pr\{T^l\leq t\cond\traj{l-1}\}
\end{align*}
is well-defined due to the assumption on continuity. At the same time, given $T^l=t^l$, $K^l$ is a discrete random variable taking values from $E^l$, whose distribution
\begin{align*}
  \Pi^l(k^l\cond t^l) \coloneqq \Pr\{K^l=k^l \cond T^l=t^l,\traj{l-1}\}
\end{align*}
is described by a probability mass function. For instance, in Algorithm~\ref{alg:nMGAExact}, this PMF is explicitly defined to be
\begin{align}\label{eqn:nmgaPMF}
  \tilde\Pi^l(k^l\cond t^l) \coloneqq \frac{\psi_{k^l}(0\cond t^l-t^l_{k^l})}{\sum_{i\in E^l}\psi_i(0\cond t^l-t^l_i)}.
\end{align}
Combining the two, the joint density function of $T^l, K^l$ is
\begin{align}\label{eqn:algodef}
  f_{T^l, K^l}(t^l, k^l\cond \traj{l-1}) \coloneqq f_{T^l}(t^l)\cdot\Pi^l(k^l\cond t^l).
\end{align} % not that redundant since it's needed for the proof...
We emphasize that the behavior of the algorithm is completely defined by Eq~(\ref{eqn:algodef}). Let $f^A_{T^l,K^l}, f^B_{T^l,K^l}$ be the joint density functions for algorithm A and B. The two algorithms are equivalent if for any $l>0$, given any previous trajectory $\traj{l-1}$, 
\begin{align*}
  f^A_{T^l,K^l}(\cdot\cond\traj{l-1}) \equiv f^B_{T^l,K^l}(\cdot\cond\traj{l-1}),
\end{align*}
which trivially implies that the joint distributions of their trajectories are also the same. 
\begin{theorem}\label{thm:equiv}
  Given Assumption 1, Algorithm~\ref{alg:nrm} and Algorithm~\ref{alg:nMGAExact} are equivalent.
\end{theorem}
The proof is given in \nameref{S4_Appendix}.

Algorithm~\ref{alg:nMGAExact} is rather inefficient, because generating random numbers from $\Phi^l(t) = \prod_i \left[1-\Psi_i(t\cond t^{l-1}-t^l_i)\right]$ (step (c) of Algorithm~\ref{alg:nMGAExact}) with the inverse CDF technique involves intensive root-finding. For this particular problem, we start with the initial bracket of [0.1, 1] and then keep halving the lower bound and doubling the upper bound until the two ends are of different signs. For the root-finding, we make use of Brent's method\cite{brent2013algorithms} provided by SciPy\cite{2020SciPy-NMeth}. 

Every evaluation done by the root-finding procedure scales as $\bigO(N)$ and if the desired error tolerance is $\epsilon$, it takes $\bigO(\log \frac{1}{\epsilon})$ evaluations for Brent's method to converge. 

The per-iteration time complexity of Algorithm~\ref{alg:nMGAExact} is thus $\bigO(N\log\frac{1}{\epsilon})$.

By replacing step (b) and (c) of Algorithm~\ref{alg:nMGAExact} with sampling from an exponential distribution, Algorithm~\ref{alg:nMGAApprox} gets rid of the costly root-finding. There are more rigorous ways \cite{Boguna2014} to justify the usage of an exponential distribution, but one intuitive way is to treat every interevent time as exponentially distributed with its corresponding instantaneous rate and update the rate in each iteration, as illustrated in Fig.~\ref{figHazard}. It is trivial to check that this approximation is exact when all interevent distributions happen to be exponential. Indeed, this approximation assumes the hazard function $\lambda^l_i$ to be constant in the time interval $(t^{l-1}, t^l)$, and would be quite accurate if $t^l-t^{l-1}$ happens to be small. 
\begin{figure}[!h]
  \centering
  \includegraphics[scale=0.5]{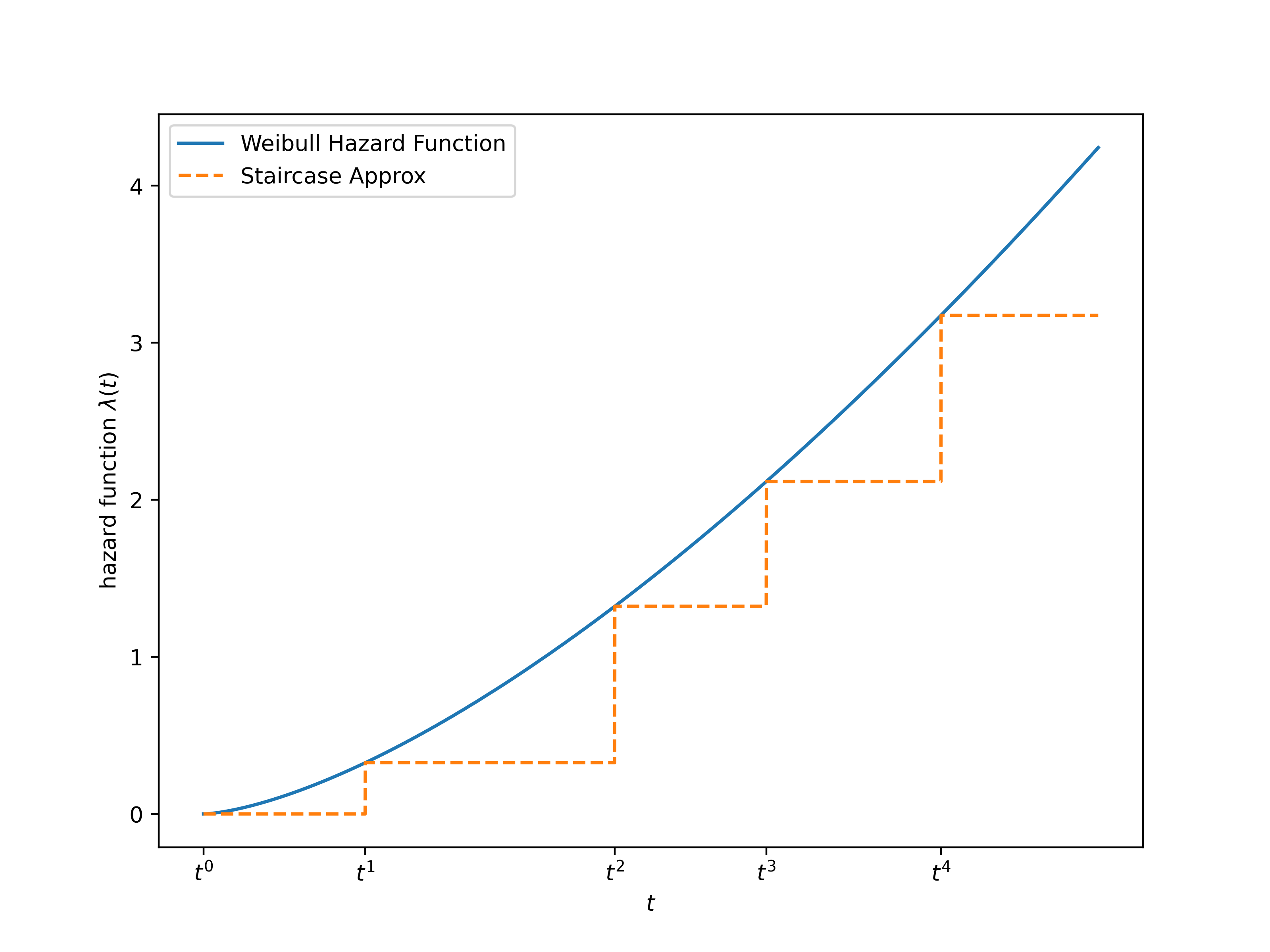} 
  \caption{
    \textbf{Staircase approximation v.s. actual hazard function.} $t^0,...,t^4$ are times of execution. The Weibull hazard function with $\kappa=1.5, \lambda=1$ is approximated by a staircase function since instantaneous rates are recomputed and updated only when an event occurs. 
  }
  \label{figHazard}
\end{figure}

We also make some observations about its error bound. Recall that we have the following equation of the hazard function $\lambda(t)$, PDF $\psi(t)$ and complementary CDF (CCDF) $\Phi(t)$:
\begin{align}\label{eqn:hazardAndPDF}
  &\lambda(t) \coloneqq \frac{\psi(t)}{\Phi(t)} = -\dv{}{t}\ln\Phi(t)\implies\int^t_0 d\tau\lambda(\tau) = -\ln(\Phi(t)).
\end{align} 
The hazard function of event emitter $i$ is $\lambda_i(t)\coloneqq\frac{\psi_i(t)}{\Phi_i(t)}$. Since in Algorithm~\ref{alg:nMGAApprox} the hazard function is recomputed only at $t^0, t^1, t^2\dots$, we denote the staircase hazard function of event emitter $i$, $\bar\lambda_i(t)$, as 
\[\bar\lambda_i(t) \coloneqq \lambda_i(t^l) \text{ if } t^l \leq t < t^{l+1}.\]
As a result, we have
\[-\ln(\Phi_i(t)) = \int^t_0 d\tau\lambda_i(\tau) = \epsilon + \int^t_0 d\tau\bar\lambda_i(\tau) = \epsilon + \sum_{l=0}^{t^l < t} \lambda_i(t^l)\cdot\Delta t^l\] 
where 
\[\Delta t^l = \begin{cases}
  t^{l+1} - t^l &\text{ if } t^{l+1} < t\\
  t - t^l &\text{ otherwise.}
\end{cases}\] 
Note that the error term $\epsilon$ is the difference between the integral $\int^t_0 d\tau\lambda_i(\tau)$ and its left Riemann sum. As a sanity check, note that $\epsilon=0$ if $\lambda_i(t)$ is constant(as in exponential distribution). Because $\{t^0, t^1, t^2, ... t\}$ form a partition of the interval $[0, t]$, $\epsilon$ goes to zero as the mesh, $\max_l\Delta t^l$, goes to zero. If we assume for simplicity that $\Delta t^l \equiv \Delta t$ and $M_1\coloneqq \sup_{\tau\in[0, t]} |\lambda_i'(\tau)|$, then $\epsilon\leq M_1 t \Delta t$ due to the choice of left Riemann sum.

As for complexity, for each iteration, all instantaneous rates $\lambda^l_i$ need to be recomputed and summed up, which is $\bigO(N)$; generating $k$ from the PMF is also $\bigO(N)$ if standard inverse-CDF is adopted.

The per-iteration time complexity of Algorithm~\ref{alg:nMGAApprox} is thus $\bigO(N)$.

\subsection{Markovian Epidemics with Heterogeneous Infection Rates}
In the setting of this section, we assume recovery rates to be homogeneous with rate $\gamma$ on any infected node. However, per-edge infection rates are not necessarily homogeneous, in the sense that the per-edge infection rate from node $i$ to node $j$ should be indexed as $\beta_{i\to j}$. If we assume per-edge infections to be mutually independent as we do in previous sections, the total rate of infection on a susceptible node $j$, $\beta_j$, is expressed as 
\begin{align*}
  \beta_j\coloneqq \sum_{i\in\N(j)}^{i\text{ infected}} \beta_{i\to j}.
\end{align*}
The authors of \cite{kiss2017mathematics} consider this setting and propose their variant of Gillespie algorithm. Because this algorithm involves computing total rates of infection on all susceptible nodes, we call it node-centric Gillespie algorithm~(node-centric GA) for disambiguation, outlined in Algorithm~\ref{alg:nodecentric}. 

\begin{algorithm}
  \caption{Node-centric GA}\label{alg:nodecentric}
  \begin{enumerate}
    \item For each susceptible node $j$, compute $\beta_j$ by 
    \[\beta_j = \sum_{i\in\N(j)}^{i\text{ infected}} \beta_{i\to j}.\]
    \item Find total rate of recovery $\lambda_R$ and total rate of infection $\lambda_I$ by
    \[\lambda_R = \gamma\cdot(\text{\# of infected nodes}),\;\;\;\; \lambda_I = \sum_{\text{susceptible }j} \beta_j.\]
    \item Time to next event is $t\sim\text{Exp}(\lambda_R+\lambda_I)$.
    \item Generate $u\sim\text{Uniform}(0, 1)$. If $u<\lambda_R / (\lambda_R+\lambda_I)$, then the next event is a recovery; otherwise, an infection.
    \item If the next event is a recovery, choose an infected node uniformly at random to recover and update the infection rates of its neighbors.
    \item If the next event is an infection, choose a susceptible node to infect according to the PMF $p_j = \beta_j / \lambda_I$ and update the infection rates of its neighbors.
    \item Go back to step 1.
  \end{enumerate}
\end{algorithm}

While all other steps can be optimized to $\bigO(1)$ via careful tracking, step 6 in Algorithm~\ref{alg:nodecentric} is the obvious bottleneck, since sampling from a PMF is $\bigO(N)$. 

The per-iteration complexity of Algorithm~\ref{alg:nodecentric} is thus $\bigO(N)$. 

\subsection{Markovian Epidemics with Homogeneous Infection Rates}
Note that in Algorithm~\ref{alg:nodecentric}, it takes little effort to locate the infected node to recover in step 5 of Algorithm~\ref{alg:nodecentric}, since recovery rates are assumed to be homogeneous and sampling uniformly at random is $\bigO(1)$. If we go a step further and assume all per-edge infection rates to be $\beta$, i.e., $\forall i, j, \beta_{i\to j}\equiv\beta$, we can optimize away the costly step 6 of Algorithm~\ref{alg:nodecentric}. The idea of grouping event emitters with the same rate is a natural one, and it is not surprising if many have been doing this optimization already without realizing it. Indeed, in \cite{ferreira2012epidemic}, the authors adopt this idea for their discretized simulations. We present in this section a version of the algorithm that supports continuous-time simulations as well as some tips for its implementation. For the rest of the article, we will refer to this algorithm as optimized GA. 

This optimization involves maintaining a hashset of infected nodes $S_R$ and another hashset of active edges (edges connecting an infected node to a susceptible one) $S_I$. This hashset data structure should support the following operations:
\begin{itemize}
  \item $\bigO(1)$ insertion.
  \item $\bigO(1)$ lookup \& deletion.
  \item $\bigO(1)$ choose an element uniformly at random.
\end{itemize}
While the first two are standard hashset operations, the third requires some engineering. A possible implementation is to maintain a hashtable that maps the elements into the indices of a random-access array. The algorithm is stated in Algorithm~\ref{alg:OptGA}. 
\begin{algorithm}
  \caption{Optimized GA}\label{alg:OptGA}
  \begin{enumerate}
    \item Find total rate of recovery $\lambda_R$ and total rate of infection $\lambda_I$ by 
    \[\lambda_R = |S_R|\cdot\gamma, \;\;\;\;\lambda_I = |S_I|\cdot\beta.\]
    \item Time to next event is $t\sim\text{Exp}(\lambda_R+\lambda_I)$.
    \item Generate $u\sim\text{Uniform}(0, 1)$. If $u<\lambda_R / (\lambda_R+\lambda_I)$, then the next event is a recovery; otherwise, an infection.
    \item If the next event is a recovery (infection), choose an element from $S_R$($S_I$) uniformly at random to determine the specific location of the recovery (infection).
    \item Update the two hashsets according to Table~\ref{tableSIS}.
    \item Go back to step 1.
  \end{enumerate}
\end{algorithm}

Because each individual step is $\bigO(1)$, this optimized Gillespie algorithm indeed brings the per-iteration cost down to $\bigO(1)$.

Even though we have only shown how to simulate an SIS epidemic, this optimizing technique of grouping event emitters with the same Poisson rate can be applied to all compartment-based network models with homogeneous, constant transition rates. As long as the number of groups does not exceed $\bigO(\log N)$, optimized GA will outperform NRM. However, it does not work for more intricate scenarios where infection rates are allowed to be heterogeneous, either following a certain distribution\cite{qu2017sis}\cite{ferguson2005strategies}, such as gamma or log-normal, or assigned according to empirical contact intensity data\cite{buono2013slow}. Nor does it work for non-Markovian cases, where event rates are time-varying and far from homogeneous. 

\section{Experiments}
We have analyzed the performance of different algorithms in earlier sections using the big-O notation. The goal of this section is to study their performance in practice. 

\subsection{Non-Markovian Simulations}
In this section, we conduct two numerical simulations using the three algorithms for non-Markovian epidemics: NRM~(\ref{alg:nrm}), nMGA-Exact~(\ref{alg:nMGAExact}), and nMGA-Approx~(\ref{alg:nMGAApprox}). 

In both experiments, the network used is an Erdos-Renyi $G(n, p)$ random graph\cite{erdHos1960evolution} with mean degree $\meandeg=5$. The network, once generated, will be held fixed for all repeated runs to reduce irrelevant variance. Both experiments start with 10 initially infected nodes. The implementation largely follows Table~\ref{tableSIS}, even though we need to pay extra care to the bookkeeping of elapsed times (see \nameref{S1_Appendix}), as is also suggested in \cite{Boguna2014}. 

The epidemic model of choice is an SIS epidemic with per-edge time to infection following
$\text{Weibull}(\kappa=2, \lambda=5.641895836)$, with a mean time to infection of $\bar{T}_{\text{inf}} = 5$. The time to recovery follows $\text{Exp}(0.25)$, with a mean time to recovery of $\bar{T}_{\text{rec}} = 4$.

The effective infection rate, as defined for the general non-Markovian case, is $\mathcal{R} = \langle k\rangle\frac{\bar{T}_{\text{rec}}}{\bar{T}_{\text{inf}}} = 5\cdot \frac{4}{5} = 4$,
which is way above 1 and guarantees an endemic state. 

In both experiments, the trajectories of each algorithm are collected and linearly interpolated. These interpolated trajectories are then evaluated on equally spaced timestamps, allowing us to compute the sample mean standard deviation for each timestamp, with which we produce the error bars. 

The first setup simulates an SIS epidemic on a network with 1000 nodes and trajectories are produced by running 10000 iterations. 150 repeated runs are conducted for each algorithm. The trajectories are visualized in Fig.~\ref{figSIS1000}. Performance benchmarking is summarized in Table~\ref{tableSISPerf}, showing that NRM is both statistically exact \textbf{and} performant. The row ``$\Delta\sim\Phi$'' stands for the time spent generating times to next event, and the row ``$K\sim\Pi$'' stands for the time spent choosing the event emitter. 

The second setup simulates an SIR epidemic with the same interevent distributions on a graph with 10000 nodes, and the simulation continues until there is no infected node left. The trajectories are visualized in Fig.~\ref{figsir}. Performance benchmarking is summarized in Table~\ref{tableSIRPerf}.

\begin{table}[H]
	\centering
	\caption{\textbf{Performance benchmarking: SIS epidemic + 1000 nodes.}}
	\begin{tabular}{|c|c|c|c|}
	  \hline
	  & nMGA-Exact      & nMGA-Approx   & NRM            \\ \hline
	Total & $\sim$26 mins   & $\sim$65 secs    & $\sim$0.8 secs \\ 
	\hline
	$\Delta\sim\Phi$ & $\sim$25 mins & $\sim$30 secs & \diagbox[width=6em]{}{} \\ 
	\hline
	$K\sim\Pi$  & $\sim$30 secs   & $\sim$30 secs & \diagbox[width=6em]{}{} \\ 
	\hline
	\end{tabular}
	\label{tableSISPerf}
\end{table}

\begin{table}[H]
	\centering
	\caption{\textbf{Performance benchmarking: SIR epidemic + 10000 nodes.}}
	\begin{tabular}{|c|c|c|c|}
	  \hline
	  & nMGA-Exact      & nMGA-Approx   & NRM            \\ \hline
	Total & $\sim$80 mins   & $\sim$150 secs    & $\sim$0.8 secs \\ 
	\hline
	$\Delta\sim\Phi$ & $\sim$80 mins & $\sim$75 secs & \diagbox[width=6em]{}{} \\ 
	\hline
	$K\sim\Pi$  & $\sim$50 secs   & $\sim$75 secs & \diagbox[width=6em]{}{} \\ 
	\hline
	\end{tabular}
	\label{tableSIRPerf}
\end{table}

\begin{figure}[H]
  \centering
  \includegraphics[scale=0.7]{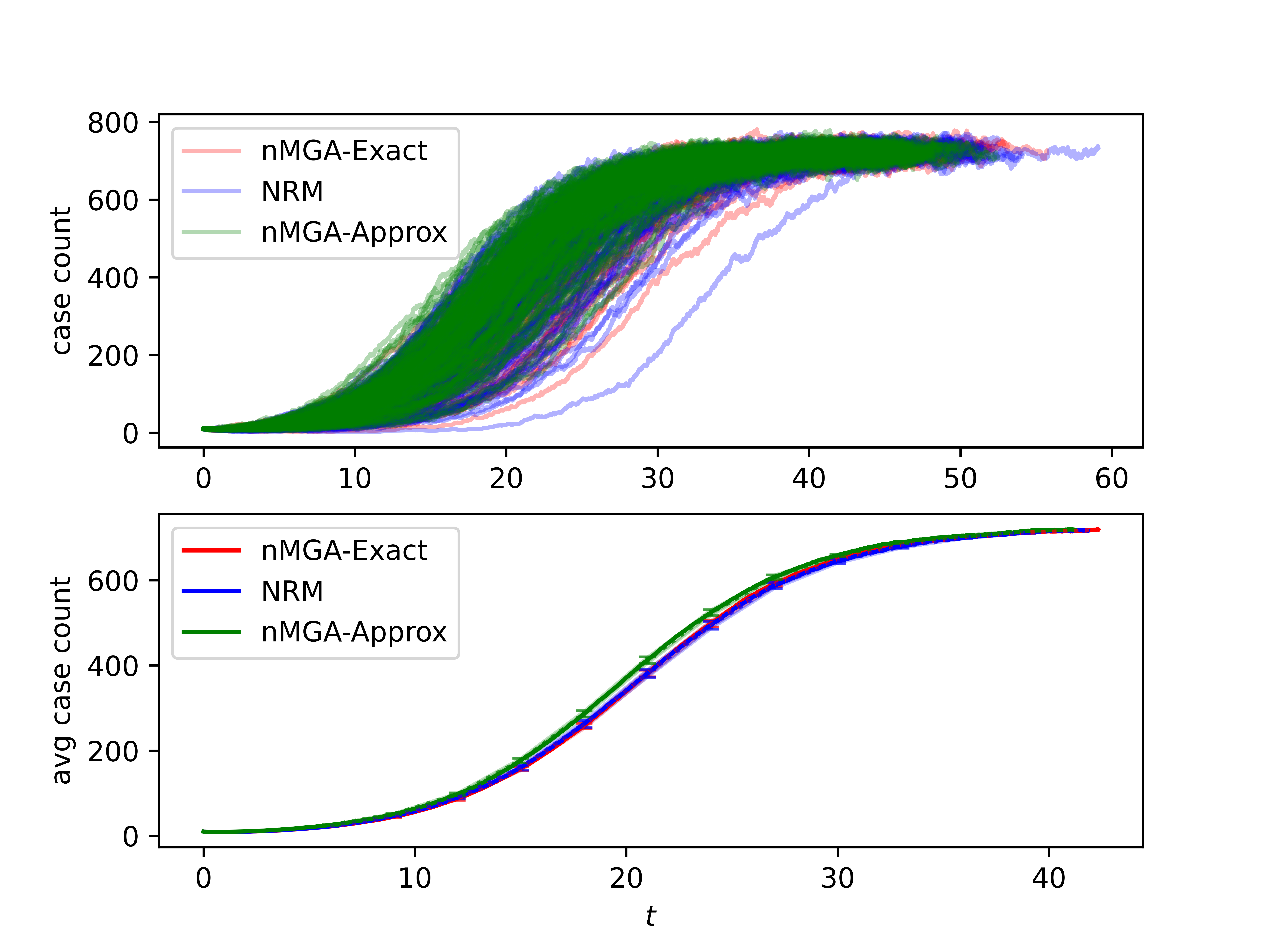} 
  \caption{{\bf Trajectories and average trajectories: SIS epidemic + 1000 nodes.} The x-axis is time and y-axis is the number of infected nodes in the network. In the lower half of Fig.~\ref{figSIS1000}, the green curve standing for the average trajectory of nMGA-Approx diverges from the other two, which is expected from an approximate algorithm. Meanwhile, the average trajectories of NRM and nMGA-Exact are indistinguishable, which is also expected since the two are statistically equivalent. }
  \label{figSIS1000}
\end{figure}

\begin{figure}[H]
	\centering
	\includegraphics[scale=0.7]{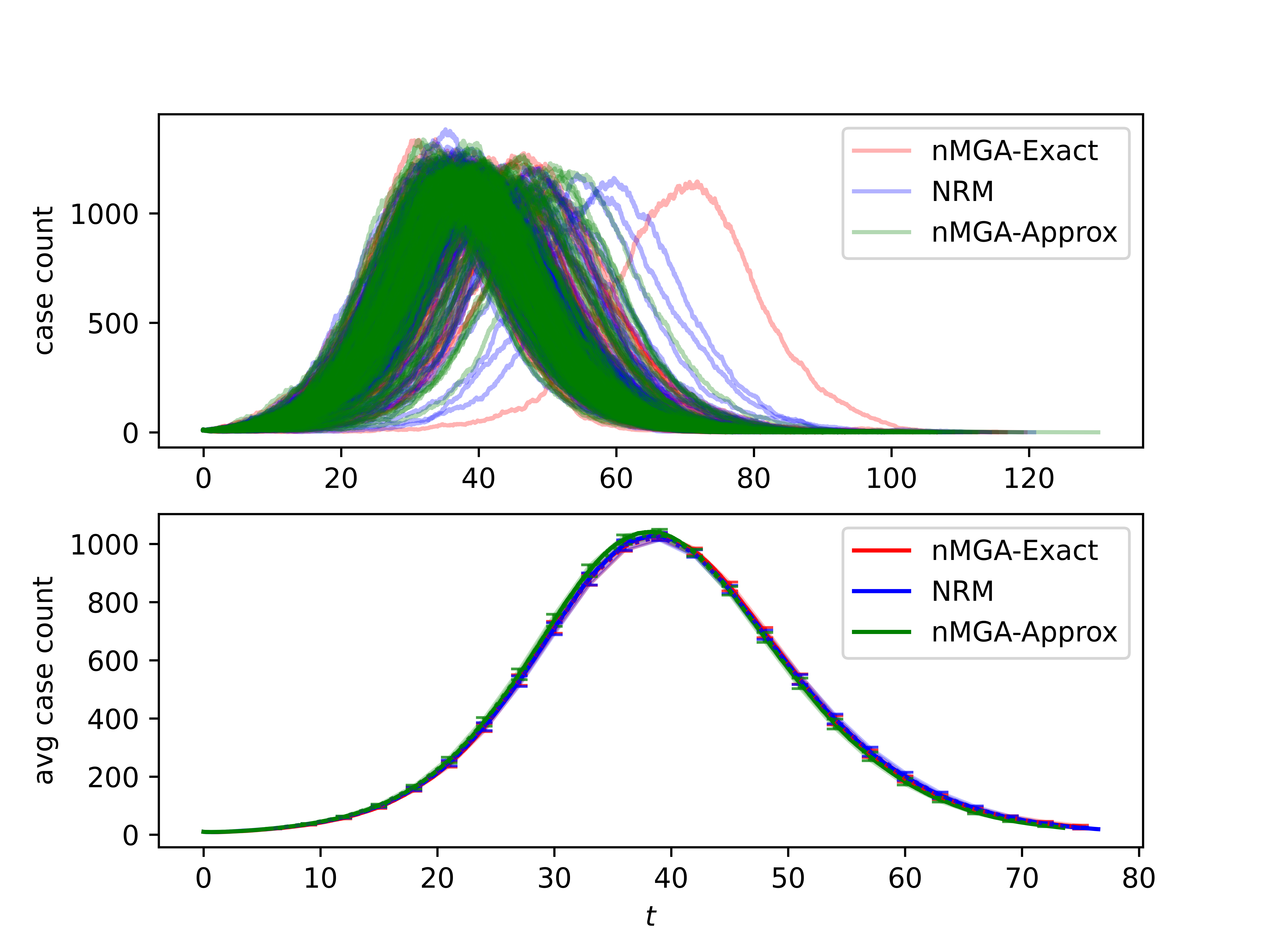} 
	\caption{{\bf Trajectories and average trajectories: SIR epidemic + 10000 nodes.} We observe that NRM~(blue) and nMGA-Exact~(red) are once again indistinguishable, while nMGA-Approx~(green) predicts a slightly earlier peak. }
	\label{figsir}
\end{figure}

\subsection{Markovian Simulations}
In this section, we conduct numerical simulations on different graph sizes using the three algorithms for Markovian epidemics with homogeneous infection rates: node-centric GA~(\ref{alg:nodecentric}), NRM~(\ref{alg:nrm}), and optimized GA~(\ref{alg:OptGA}).
To study how the three algorithms scale, we use ER graphs of different sizes, each with 5\% initially infected nodes. For each graph size, 20 runs of each algorithm are carried out to measure average time consumption. For each run, the simulation continues until a 30-day horizon in the simulated world is met. Per-edge infection time and recovery time are configured to be
\[T_\text{inf}\sim \text{Exp}(0.125),\;\;\;\; T_\text{rec}\sim \text{Exp}(0.25).\]

Note that for Markovian epidemics, when a horizon is given, the number of iterations needed to reach said horizon scales linearly with the size of the graph. Suppose the given horizon is $H$ and we have $N$ self-renewing event emitters with rate $\lambda$. The time between two consecutive events $T$ follows an exponential distribution with rate $N\lambda$, with $\mathbf{E}[T] = (N\lambda)^{-1}$. The expected number of iterations needed to reach horizon $H$ is then $\frac{H}{\mathbf{E}[T]} = HN\lambda$. 

The performance of the three algorithms is shown in Fig.~\ref{figblip}.

\begin{figure}[H]
  \centering
  \includegraphics[scale=0.60]{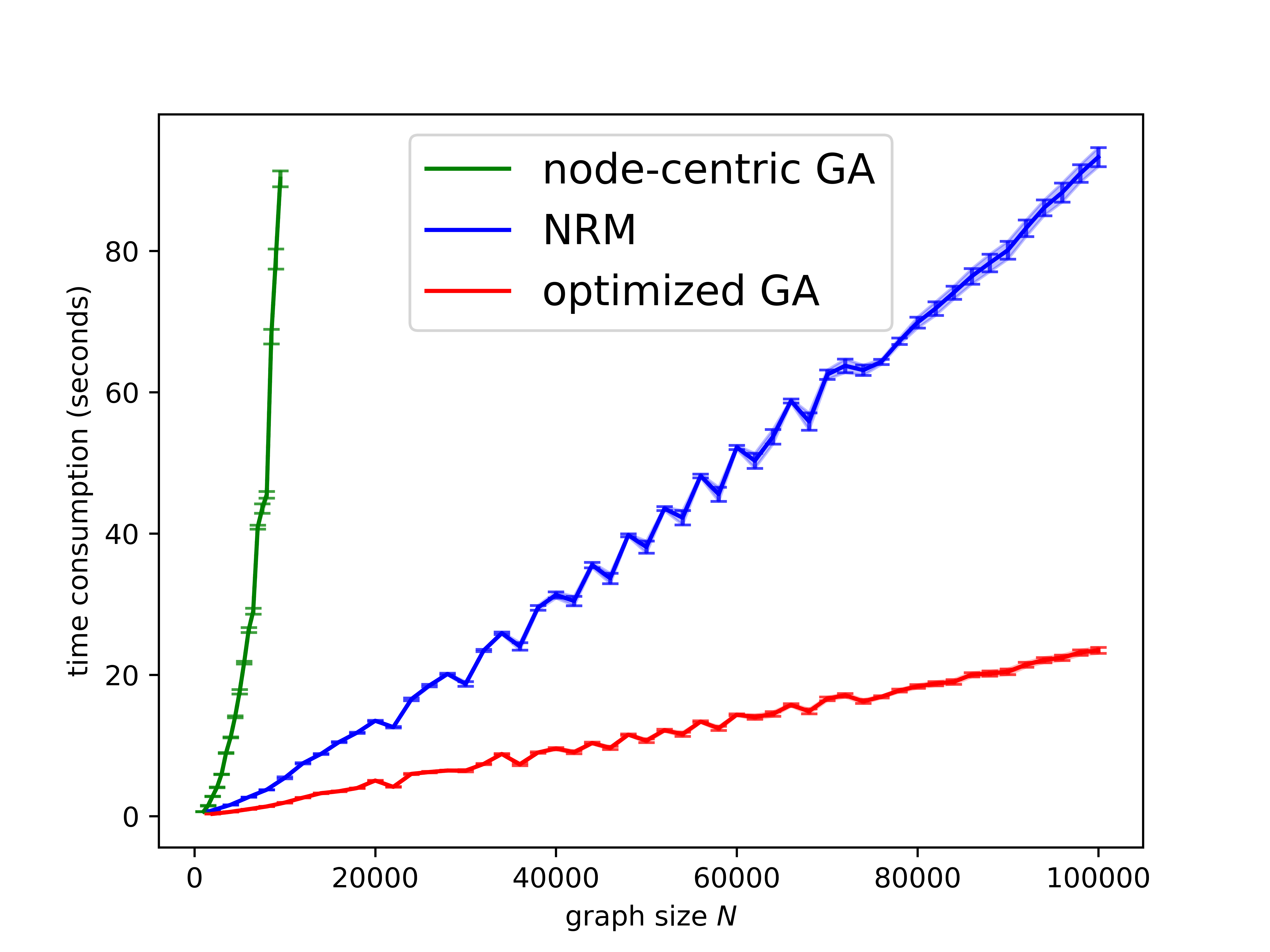} 
  \caption{{\bf Markovian epidemics with homogeneous rates.} As predicted by our complexity analysis, for optimized GA, a per-iteration cost of $\bigO(1)$ implies an overall cost of $\bigO(N)$. Meanwhile, NRM is around 3x-4x slower for the graph sizes considered. Node-centric GA, however, already borders on the intractable for large graph sizes.}
  \label{figblip}
\end{figure}

\subsection{Temporal/Adaptive Networks}
Recently, researchers have shown interest in temporal/adaptive networks that arise from both epidemiology\cite{gross2006epidemic}\cite{wang2019coevolution} and beyond\cite{gross2008adaptive}. New algorithms have been developed for stochastic simulations\cite{Vestergaard2015}. Theoretical efforts have also been made to unify this emerging field\cite{wang2019coevolution}. 

A temporal network refers to a network whose connectivity changes over time. On the one hand, this change in connectivity can be externally driven, in that the variation of network structure is independent from the dynamical processes taking place within the network. For example, the commuting pattern in a contact network is arguably externally driven. Externally driven network changes can also be an artifact of contact tracing/surveillance, if such data is collected periodically. On the other hand, it can also be internally driven, in that the network adapts to the dynamical processes through time. For instance, a susceptible node neighboring an infected one may want to sever the edge between them. 

\subsubsection{An Adaptive Network with Smart Rewiring}
In this section, we adopt a simple rewiring mechanism specified in \cite{gross2006epidemic} to simulate an SIS epidemic on an internally driven adaptive network. Gross et al.\cite{gross2006epidemic} choose discretized time steps for simulations. For each time step, each infected node recovers with probability $r$, each active edge transmits disease with probability $p$, and each active edge gets rewired with probability $w$. The goal of this section is to reproduce some of their results using continuous-time simulation. 

We get started by defining the event emitters. In plain words, rewiring means that a susceptible node disconnects from its infected neighbor, and then reconnect to a randomly chosen \textbf{susceptible} node in the network, while making sure the network remains a simple graph. We define the Rewiring event emitter in Table~\ref{tableSISRewire}.

\begin{table}[H]
  \centering
  \caption{\textbf{Event emitters for smart-rewiring.}}
  \begin{tabular}{|c|c|}
    \hline
       & Rewiring \\ 
       \hline
    ID & Rewire($x\to y$) \\ 
    \hline
    f  & \makecell{edge $(x,y)$ removed;\\ node $z$.state=S randomly chosen; \\ edge $(y,z)$ added;} \\ \hline
    C  & $\emptyset$\\ \hline
    R  & $\{\Inf(x\to y), \Rewire(x\to y)\}$\\ \hline
    \end{tabular}
  \label{tableSISRewire}
\end{table}

We also modify existing Infection and Recovery emitters accordingly. Essentially, an Infection emitter $\Inf(x\to y)$ and a Rewiring emitter with the same signature $\Rewire(x\to y)$ should always be created and removed together. The modified emitters are shown in Table~\ref{tableSISRewireMod}.
\begin{table}[H]
  \centering
  \caption{\textbf{Infection and recovery event emitters for an SIS epidemic with smart-rewiring.}}
  \begin{tabular}{|c|c|c|}
  \hline
             & Infection                          & Recovery                      \\ \hline
  ID         & Inf($x\to y$) & Rec($x$)    \\ \hline
  f          & $y$.state $\gets$ I (Infected)         & $x$.state $\gets$ S (Susceptible) \\ \hline
  C          & \makecell{$\{\Rec(y)\}$\\$\cup$\\$\{\Inf(y\to z), \Rewire(y\to z)$\\ $\;|\; z\in \N(y), z.\text{state=S}\}$} & \makecell{$\{\Inf(z\to x), \Rewire(z\to x)$ \\ $\;|\; z\in \N(x), z\text{.state=I}\}$} \\ \hline
  R          & \makecell{$\{\Inf(z\to y), \Rewire(z\to y)$ \\ $\;|\; z\in \N(y), z\text{.state=I}\}$} & \makecell{$\{\Rec(x)\}$ \\ $\cup$ \\$\{\Inf(x\to z), \Rewire(x\to z)$\\ $\;|\; z\in \N(x),z\text{.state=S}\}$} \\ \hline
  \end{tabular}
  \label{tableSISRewireMod}
\end{table}

To reproduce the results in \cite{gross2006epidemic}, we assume all transitions to be exponential with homogeneous rates, making optimized GA (\ref{alg:OptGA}) applicable. We denote the recovery rate as $\gamma$, per-edge infection rate as $\beta$, and rewiring rate as $\omega$. 

A consequence of this rewiring mechanism is the emergence of a hysteresis loop\cite{gross2006epidemic}. To reproduce this phenomenon, we consider two different initial conditions:
\begin{itemize}
  \item An Erdos-Renyi $G(N, p)$ network with 1\% initially infected nodes.
  \item A heavily rewired network where the epidemic has been established.
\end{itemize} 
To produce the second initial condition, we simply start the simulation from the first one, and save the network state once the system reaches its endemic equilibrium. From the first initial condition, we attempt different infection rates and record the epidemic size at equilibrium in order to find the invasion threshold, i.e., the lowest infection rate required for a new epidemic to pervade the population. From the second one, we do the same to find the persistence threshold, i.e., the lowest infection rate required for an established epidemic to persist. The results are summarized in Fig.~\ref{fighys}.
\begin{figure}[H]
  \centering
  \includegraphics[scale=0.45]{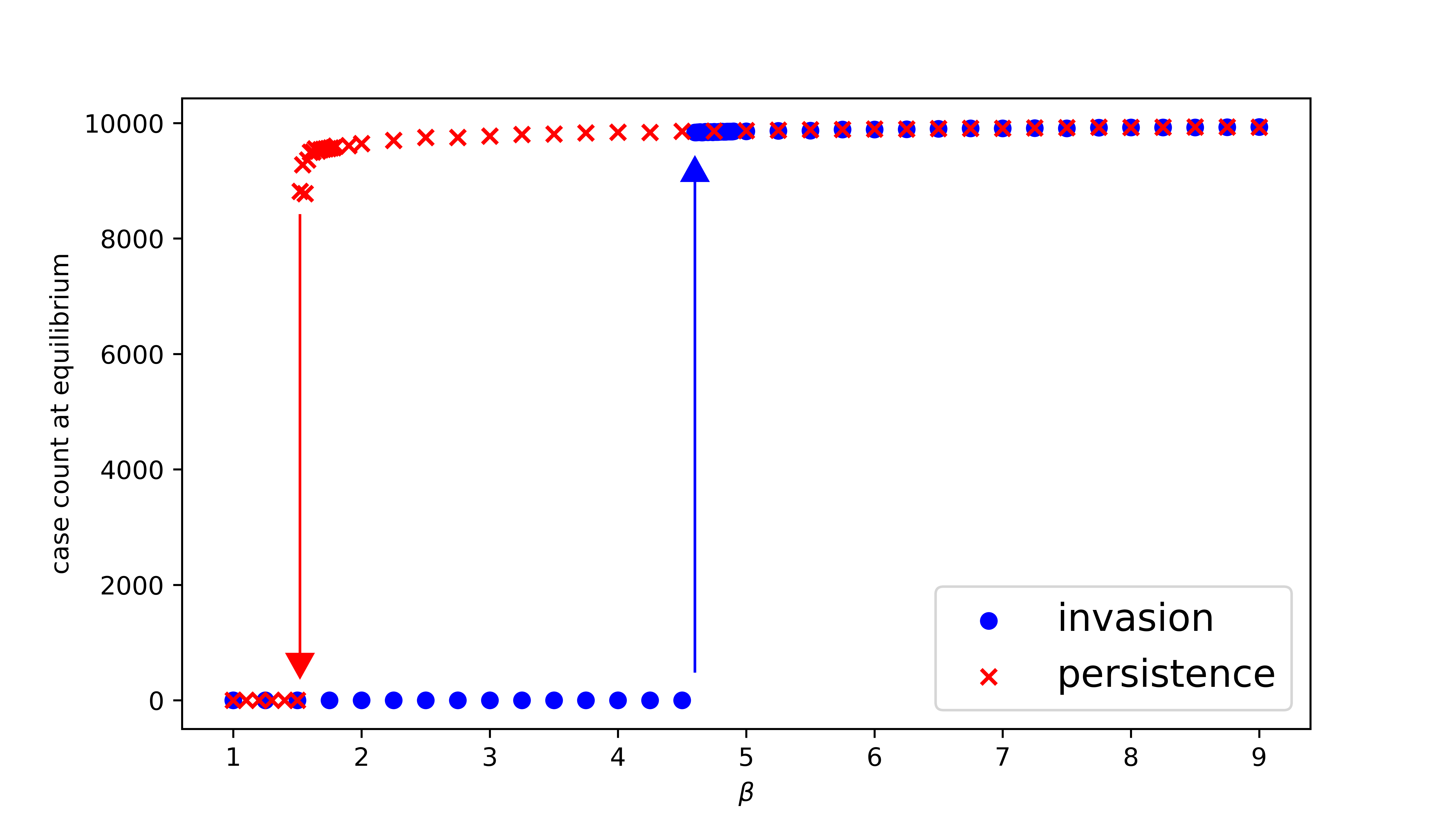} 
  \caption{{\bf Hysteresis loop with rewiring mechanism.} $\gamma=1$, $\omega=100$, $N=10^4$, mean degree$=20$. The blue dots are produced by running the simulation with different infection rates starting from the first initial condition. The red dots are produced by starting from the second initial condition. For both, we record the average case count at equilibrium. The average case count at equilibrium is set as zero if more than 90\% of repeated runs terminate prematurely due to the epidemic dying out before the given number of iterations is finished. Otherwise, the average is computed as the average of all surviving runs. Invasion threshold is around 4.6. Persistence threshold is around 1.52. These numbers are close to the results in \cite{gross2006epidemic} after rescaling.}
  \label{fighys}
\end{figure}

According to \cite{gross2006epidemic}, with high rewiring rate (say, with $\omega=300$), if we start from the second initial condition, as $\beta$ decreases, the original equilibrium loses its stability and the system enters a stable limit cycle. This phenomenon is known as a Hopf bifurcation\cite{strogatz2018nonlinear}. We manage to produce this phenomenon and illustrate it in Fig.~\ref{figHopf}. 
\begin{figure}[!h]
  \centering
  \includegraphics[scale=0.26]{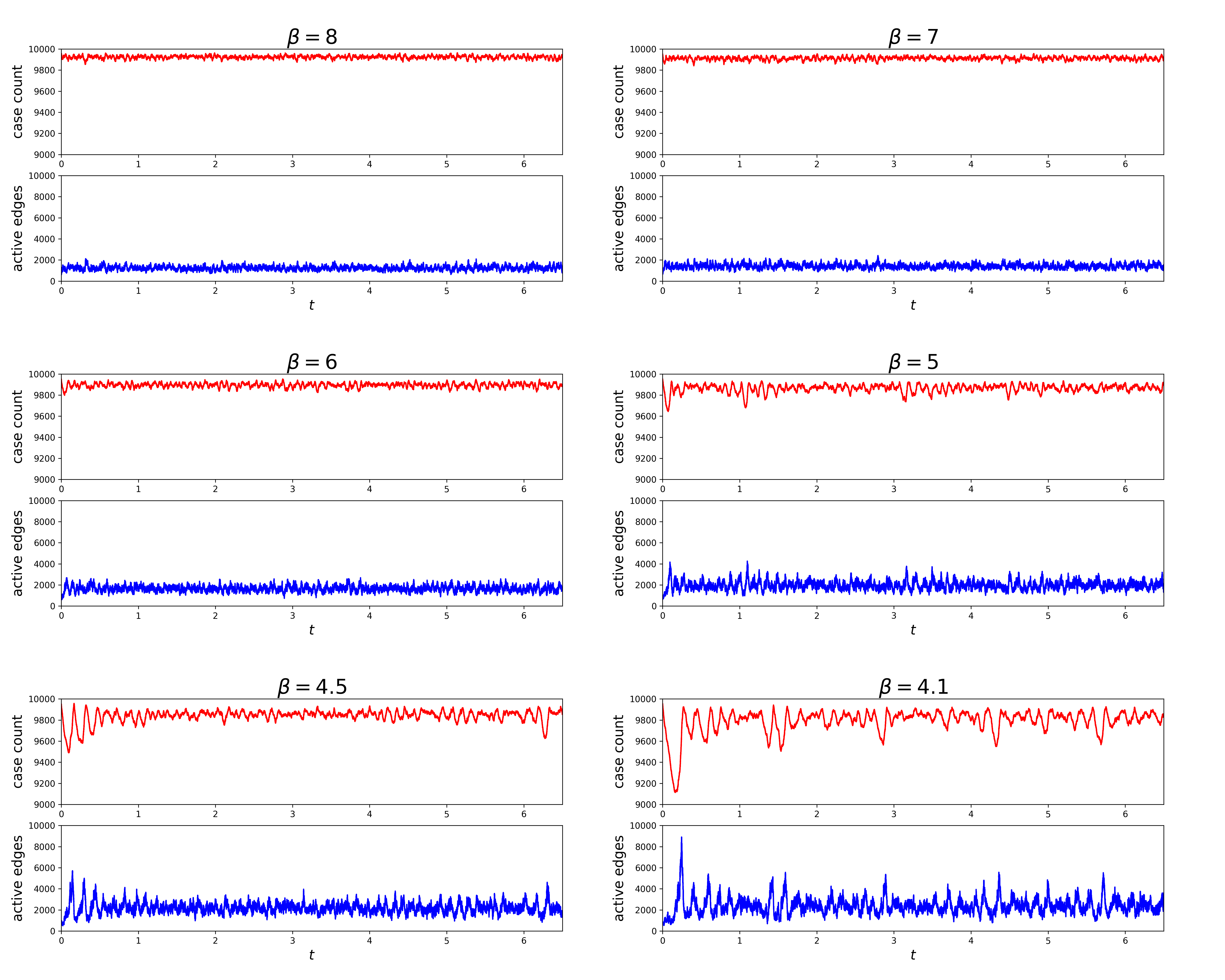}
  \caption{\textbf{Hopf bifurcation with high rewiring rate.} All trajectories start from the same heavily rewired, disease-ridden network with $\gamma=1, \beta=16, \omega=300$. Persistence threshold is around 4.1. With the rewiring rate fixed at 300, the oscillatory behavior becomes more and more outstanding as the infection rate $\beta$ decreases. }
  \label{figHopf}
\end{figure}

\subsubsection{A Temporal Network with Commuting}
In this section, we simulate an SIS epidemic on a externally driven temporal network, that is, the commuting pattern between home and school. The model we propose is inspired by the hybrid approach adopted by many schools and workplaces during the COVID-19 pandemic, where only a fraction of the personnel may be on-site each day, while the rest have to conduct their study/work at home. We carry out two sets of experiments. We first measure the performance of three algorithms (NRM~(\ref{alg:nrm}), node-centric GA~(\ref{alg:nodecentric}), optimized GA~(\ref{alg:OptGA})) on graphs of different sizes. We then explore the effect of changing attendance probability (the fraction of people allowed on-site). For both experiments, we set the per-edge infection rate to be 0.03 and the recovery rate to be 0.25.

We generate a graph for ``school'' and a graph for ``home''. Both graphs are held fixed once generated. When a node leaves one network for the other one, the node will be marked as ``off'' in the first network and marked as active with its previous state~(infected/susceptible) in the other. A node marked as ``off'' in a network does not participate in the epidemic process on that network and is effectively invisible to its neighbors. This turn-off mechanism allows us to move nodes from one network to the other without altering the underlying graph data structure, leading to a more efficient implementation. Without loss of generality, all nodes are assumed to be at home when the simulation starts. After 12 hours, each node at home leaves for school with the aforementioned attendance probability, enabling the epidemic process to continue on what are effectively two subgraphs of the two graphs. After another 12 hours, all nodes at school leave for home, completing the cycle. Apart from the commuting pattern, the simulation works just like a vanilla SIS epidemic.

Implementation-wise, NRM requires minimal adaptation, since the commuting pattern can be implemented as two event emitters ``Home2School'' and ``School2Home'' that renew each other, with interevent distributions being Dirac deltas. For node-centric GA and optimized GA, we draw inspiration from \cite{Vestergaard2015} for their adaptations. In each iteration, we generate the time to next event in the classic Gillespie way. Then we compare the time to next event to the next commuting time. If the time to next event is smaller, the event, be it an infection or a recovery, will occur. Otherwise, the event will be rejected and the commuting will occur. 

For the first experiment, we select uniformly at random 5\% of nodes to be initially infected. We run each algorithm to the end of a 60-day window in the simulated world with attendance probability being 1/3. The ``home'' network is assumed to be an ER graph with mean degree 10 and the ``school'' network is assumed to be an ER graph with mean degree 100. The average time consumption is computed as the average of 5 repeated runs. The results from the first experiment are summarized in Fig.~\ref{figPerfCommute}. 

\begin{figure}[H]
  \centering
  \includegraphics[scale=0.6]{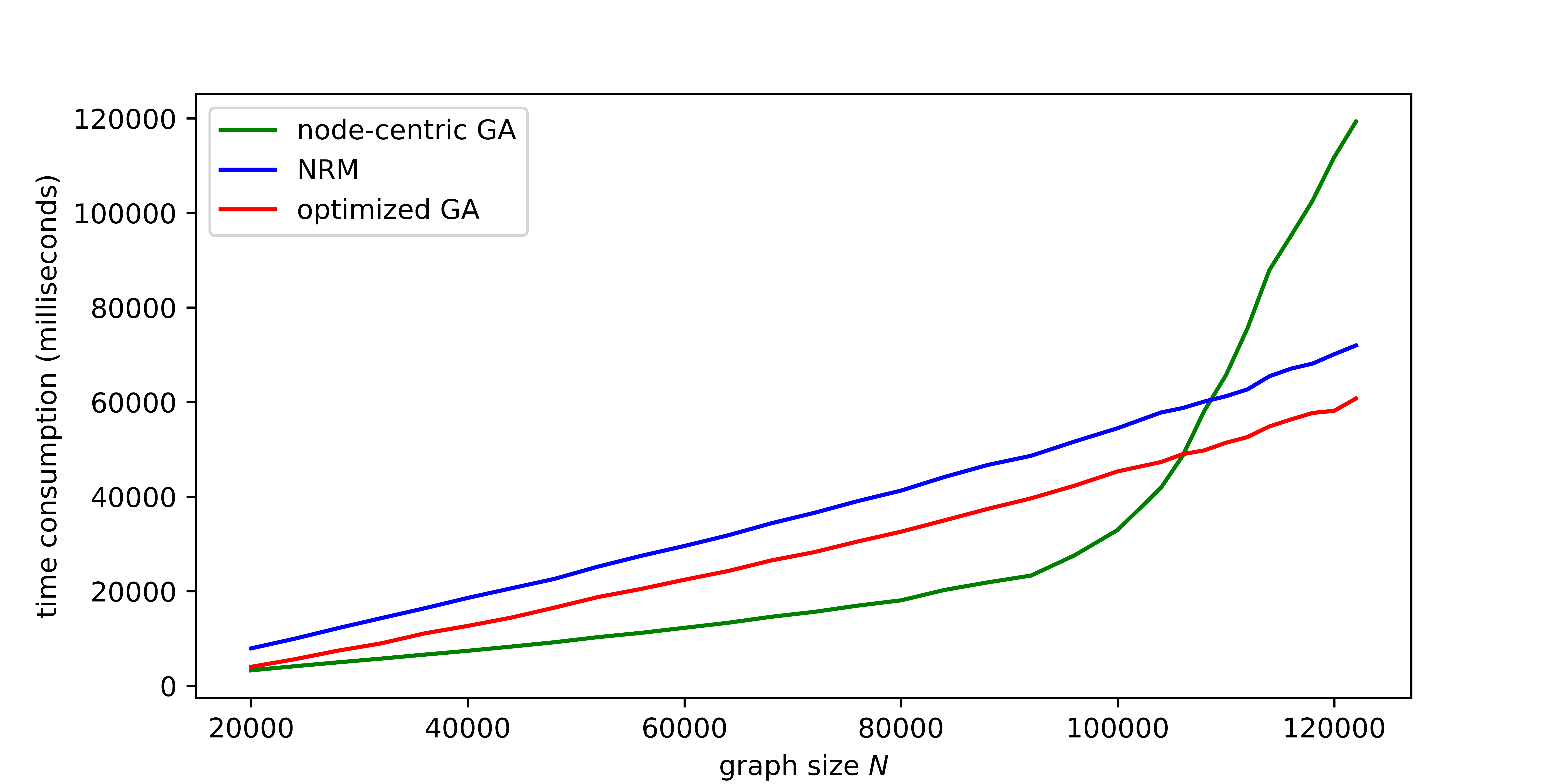}
  \caption{\textbf{Performance of the three algorithms for Markovian epidemics with homogeneous rates on graphs of different sizes.} We see NRM and optimized GA have comparable performance, although their advantage over node-centric GA is apparent only after the graph size exceeds 100000.}
  \label{figPerfCommute}
\end{figure}

To better understand the performance of the three algorithms, we also conduct instruction counting (counting the number of instructions executed by the CPU) and cache profiling (counting the number of cache hits/misses) with callgrind. The results are summarized in Fig.~\ref{figCallgrind} and Table~\ref{tableSecantSlope}.
\begin{figure}[H]
  \centering
  \includegraphics[scale=0.5]{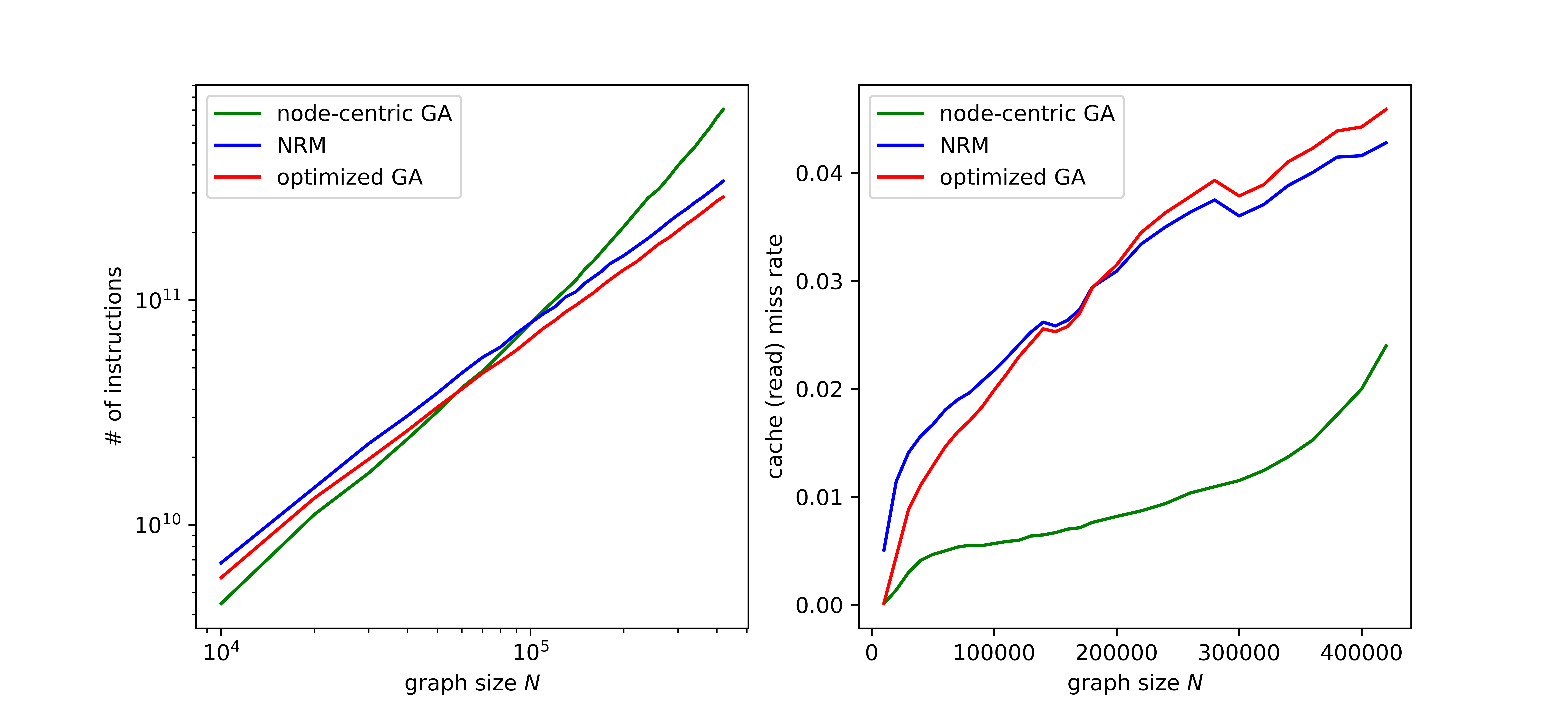}
  \caption{\textbf{Instruction counting and cache profiling.} The one to the left is a log-log plot, whose slope characterizes the complexity of the algorithm and is summarized in Table~\ref{tableSecantSlope}. The one to the right plots cache read miss rates against graph sizes, showing that as the size of the problem grows, the cache performance also deteriorates. NRM and optimized GA have worse cache performance due to the inevitable indirections in their implementation, but they eventually outperform node-centric GA despite being less cache-friendly.}
  \label{figCallgrind}
\end{figure}

\begin{table}[H]
  \centering
  \caption{\textbf{Slope of the secant line in the log-log plot (Fig.~\ref{figCallgrind}).}}
  \begin{tabular}{|c|c|c|c|}
  \hline
                  & first half & second half & overall \\ \hline
  Node-centric GA~(\ref{alg:nodecentric}) & 1.2988     & 1.6202      & 1.3544  \\ \hline
  NRM~(\ref{alg:nrm})             & 1.0481     & 1.0379      & 1.0464  \\ \hline
  Optimized GA~(\ref{alg:OptGA})    & 1.0465     & 1.0286      & 1.0434  \\ \hline
  \end{tabular}
  \begin{flushleft}
    Table notes: NRM~(\ref{alg:nrm}) and optimized GA~(\ref{alg:OptGA}) look rather linear, possibly because the problem size is still not large enough for the $\log N$ term in NRM to make a difference. The secant slope of node-centric GA goes from 1.2988 to 1.6202, which aligns with our previous discussion that its asymptotic slope is 2, since its quadratic part eventually prevails.
  \end{flushleft}
  \label{tableSecantSlope}
\end{table}

For the second experiment, the population size is fixed at 10000. The ``home'' network is assumed to be a stochastic block model with block size 5 with 1\% initially infected nodes, while the ``school'' network is assumed to be a stochastic block model with block size 50. For both networks, the edge probability $p_{ij}$ is zero if $i, j$ belong to different blocks and one if $i, j$ belong to the same block. This construction suggests that commuting serves as a form of long-distance connection. Each average trajectory is computed as the average of 10 repeated runs. The results from the second experiment are summarized in Fig.\ref{figAttendance}. 
\begin{figure}[H]
  \centering
  \includegraphics[scale=0.6]{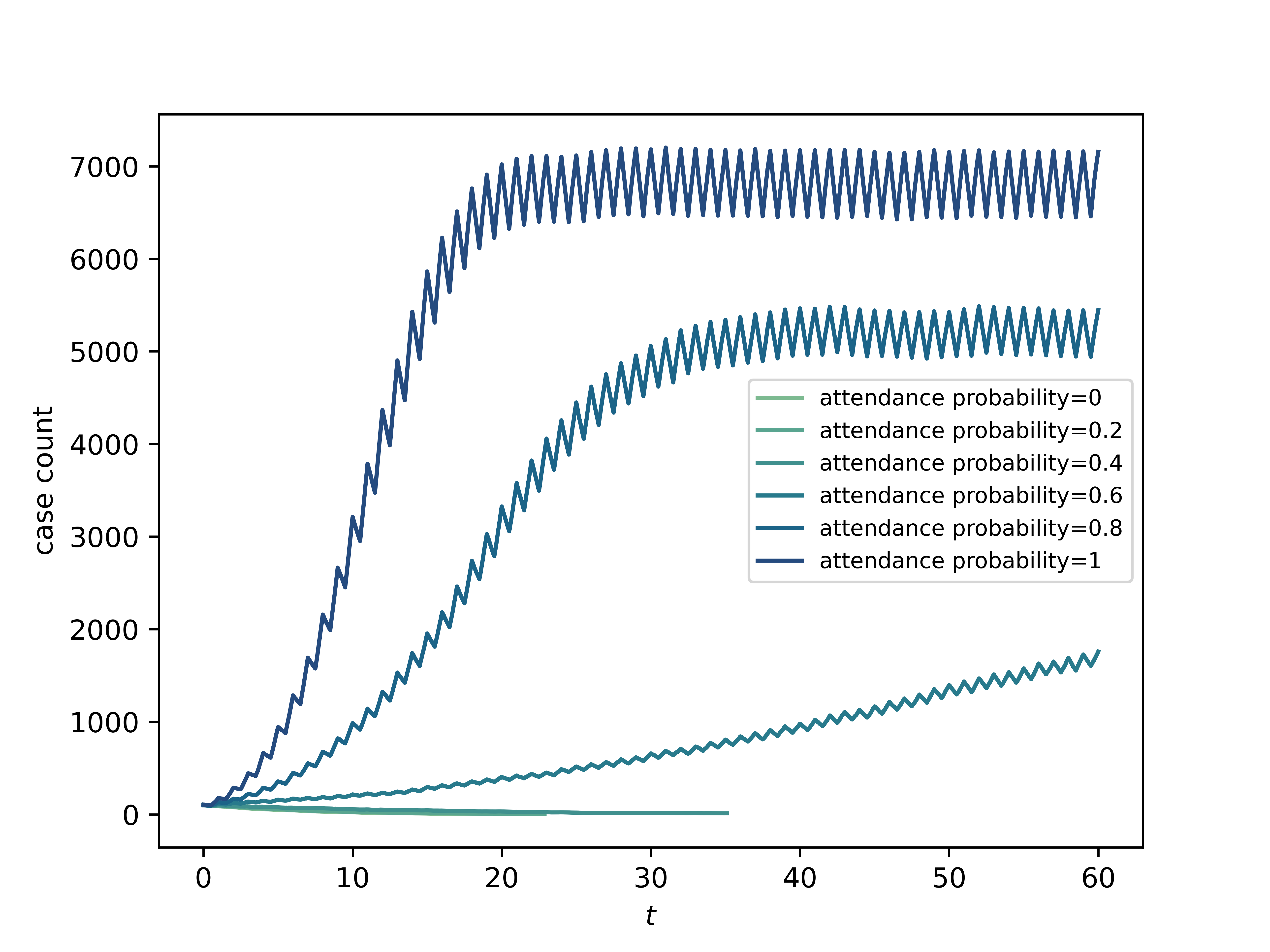}
  \caption{\textbf{Average epidemic trajectories with different attendance probabilities.} When attendance probability is 0 or 0.2, the infection rate is below the epidemic threshold. We see that, in this specific network setup, a higher attendance probability leads to a lower epidemic threshold and a larger epidemic prevalence. A higher attendance probability also gives rise to higher variance in the trajectories since the process of determining which nodes should attend the ``school'' network is randomized.}
  \label{figAttendance}
\end{figure}

\subsection{Cooperative Infections}
Up until now, we have assumed infections through edges to be independent. In this section, we review a scenario with cooperative infections mentioned in \cite{Boguna2014}, where the total rate of being infected on node $x$ is expressed as 
\[\lambda_{x, \text{total}}(t) = \left[\sum_{y\in\N(x), y.\text{state=I}} (\lambda_y(t))^{\frac{1}{\sigma}}\right]^\sigma,\]
where $\N(x)$ is the 1-hop neighborhood of vertex $x$. 

Cooperative infections are introduced in \cite{Boguna2014}. This scenario is interesting because nMGA-Approx~(\ref{alg:nMGAApprox}) can be readily applied, while it is not obvious how to apply NRM. Here in this section we show that by redesigning event emitters and generating tentative times given hazard functions, not only can we apply NRM to this scenario but we can make NRM outperform nMGA-Approx, as well. 

When $\sigma=1$, we recover the classical scenario where infections coming from neighbors are independent. Otherwise, it is nontrivial to generate interevent times when only the hazard function, in this case $\lambda_{x, \text{total}}(t)$, is specified. One approach is to utilize Eq.~\ref{eqn:hazardAndPDF}, which allows us to use the inverse CDF technique by (i) sampling $u\sim\text{Uniform}(0, 1)$ and then (ii) solve $\int^t_0 d\tau\lambda(\tau) = -\ln(u)$ for $t$. To evaluate the integral on the left-hand side, one can use a Newton-Cotes quadrature. A naive implementation is shown in Alg.~\ref{alg:newtonquadnaive}. Note that even though the midpoint rule, which has leading error term $\dt^2$, is used in Alg.~\ref{alg:newtonquadnaive}, possible solutions only take values from $\{\frac{\dt}{2} + k\dt\cond k=0,1,2\dots\}$, making the leading error term $\dt$. 
\begin{algorithm}
  \caption{\textbf{Generate interevent times given hazard function - Naive}}\label{alg:newtonquadnaive}
  \begin{algorithmic}
    \State{Generate $u\sim$ Uniform(0, 1)} 
    \State{$t\gets \dt/2$}
    \State{sum $\gets 0$}
  \While{sum $< -\ln u$}
      \State{sum $\gets \mbox{sum} + \dt\cdot \lambda(t)$}
      \State{$t\gets t + \dt$}
    \EndWhile
    \State{return $t$}
  \end{algorithmic}
\end{algorithm}

A more refined approach would involve interpolating between the left and right bound of the last interval. To be more precise, with $t_i\coloneqq \frac{\dt}{2} + i\dt$, we continue to compute the Riemann sum to iteration $n$ until
\[\sum_{i=0}^n \lambda(t_i)\dt < -\ln u < \sum_{i=0}^{n+1} \lambda(t_i)\dt, \]
meaning that the actual solution of $t$ is somewhere between $(n+1)\dt$ and $(n+2)\dt$. 

To find $t$, we adopt a simple trapezoidal interpolation scheme, illustrated in Fig~\ref{figTrapezoid}. This approach guarantees that the root-finding $\int^t_0 d\tau\lambda(\tau) = -\ln(u)$, which involves both the midpoint rule and the trapezoidal interpolation, is exactly solved if $\lambda(t)$ is linear, making the leading error term $\dt^2$. The refined algorithm is shown in Alg.~\ref{alg:newtonquadrefined}.
\begin{algorithm}
  \caption{\textbf{Generate interevent times given hazard function - Refined}}\label{alg:newtonquadrefined}
  \begin{algorithmic}
    \State{Generate $u\sim$ Uniform(0, 1)} 
    \State{$t\gets \dt/2$}
    \State{sum $\gets 0$}
    \While{true}
      \State{new\_sum $\gets \mbox{sum} + \dt\cdot \lambda(t)$}
      \If{new\_sum $< -\ln u$}
        \State{sum $\gets$ new\_sum}
        \State{$t\gets t + \dt$}
      \Else{ break}
      \EndIf
    \EndWhile
    \State{lo $\gets t-\dt/2$; $h_1\gets\lambda(\text{lo})$}
    \State{hi $\gets t+\dt/2$; $h_2\gets\lambda(\text{hi})$}
    \State{$D\gets -\ln u - \text{sum}$; $k\gets (h_2-h_1)/\dt$}
    \State{return $t + \frac{-h_1 + \sqrt{h_1^2 + 2Dk}}{k}$}
  \end{algorithmic}
\end{algorithm}

\begin{figure}[H]
  \centering
  \includegraphics[scale=0.20]{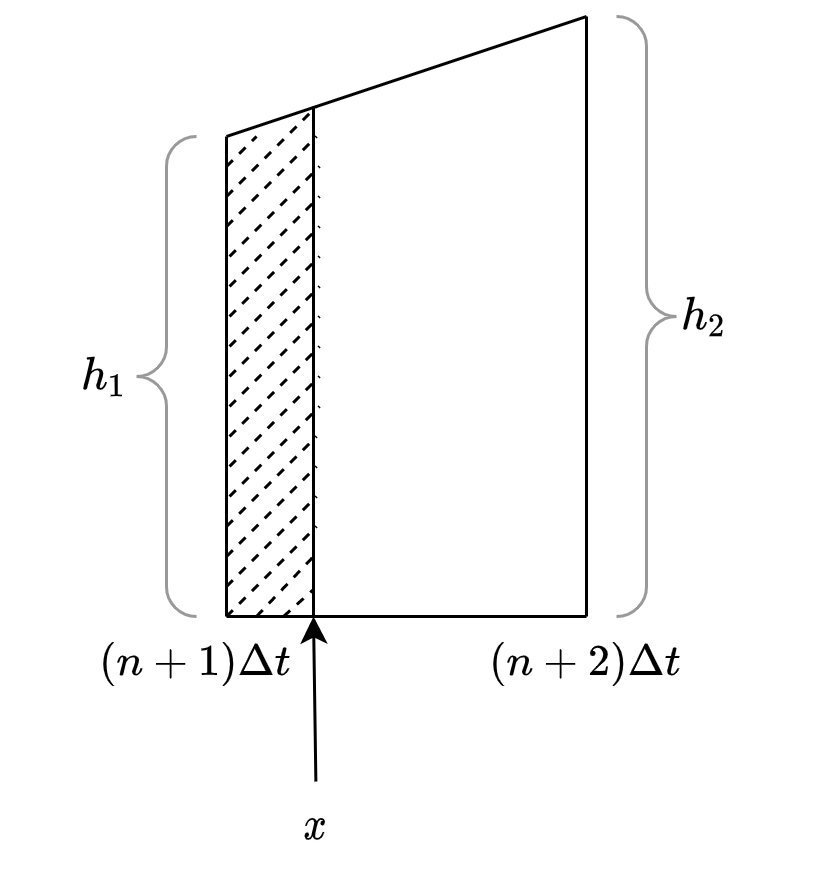} 
  \caption{{\bf Trapezoidal interpolation illustrated.} Let $D\coloneqq -\ln u - \sum_{i=0}^n \lambda(t_i)\dt > 0$, $h_1\coloneqq \lambda((n+1)\dt)$, $h_2\coloneqq \lambda((n+2)\dt)$. The problem is reduced to finding $x$ between $(n+1)\dt$ and $(n+2)\dt$ such that the shaded area is $D$. Let $k\coloneqq (h_2-h_1)/\dt$. The correct value of $x$ is thus $\frac{-h_1+\sqrt{h_1^2 + 2Dk}}{k}$.}
  \label{figTrapezoid}
\end{figure}

\noindent We also present the event emitter specification for cooperative infection in Table~\ref{tableCooperative}.
\begin{table}[H]
  \centering
  \caption{\textbf{Event emitters for cooperative infection~(SIS).}}
  \begin{tabular}{|c|c|c|}
    \hline
               & Infection                          & Recovery                      \\ \hline
    ID         & Inf($y$)                      & Rec($x$)                      \\ \hline
    f          & $y$.state $\gets$ I         & $x$.state $\gets$ S \\ \hline
    C          & \makecell{$\{\text{Rec}(y)\}$\\$\cup$\\$\{\text{Inf}(z) \;|\; z\in \N(y), z.\text{state=S}\}$} 
               & \makecell{\{Inf($x$)\}\\$\cup$\\$\{\text{Inf}(z) \;|\; z\in \N(x), z.\text{state=S}\}$} \\ \hline
    R          & \makecell{\{Inf($y$)\}\\$\cup$\\$\{\text{Inf}(z) \;|\; z\in \N(y), z.\text{state=S}\}$} 
               & \makecell{$\{\text{Rec}(x)\}$ \\ $\cup$ \\$\{\text{Inf}(z) \;|\; z\in \N(x), z.\text{state=S}\}$} \\ \hline
  \end{tabular}
  \begin{flushleft}
    Table notes: Whenever an event~(Infection/Recovery) is executed on a node $x$, for each of its susceptible neighbor $y$, the infection event emitter Inf($y$) is replaced, because the hazard function of $y$'s neighborhood has changed.
  \end{flushleft}  
  \label{tableCooperative}
\end{table}

We run simulations according to Algorithm~\ref{alg:newtonquadrefined} and Table~\ref{tableCooperative} on ER graphs of different sizes, each with a mean degree of 10 and 5\% initially infected nodes.  The results are shown in Fig.~\ref{figCooperativeTraj}. 

\begin{figure}[H]
  \centering
  \includegraphics[scale=0.40]{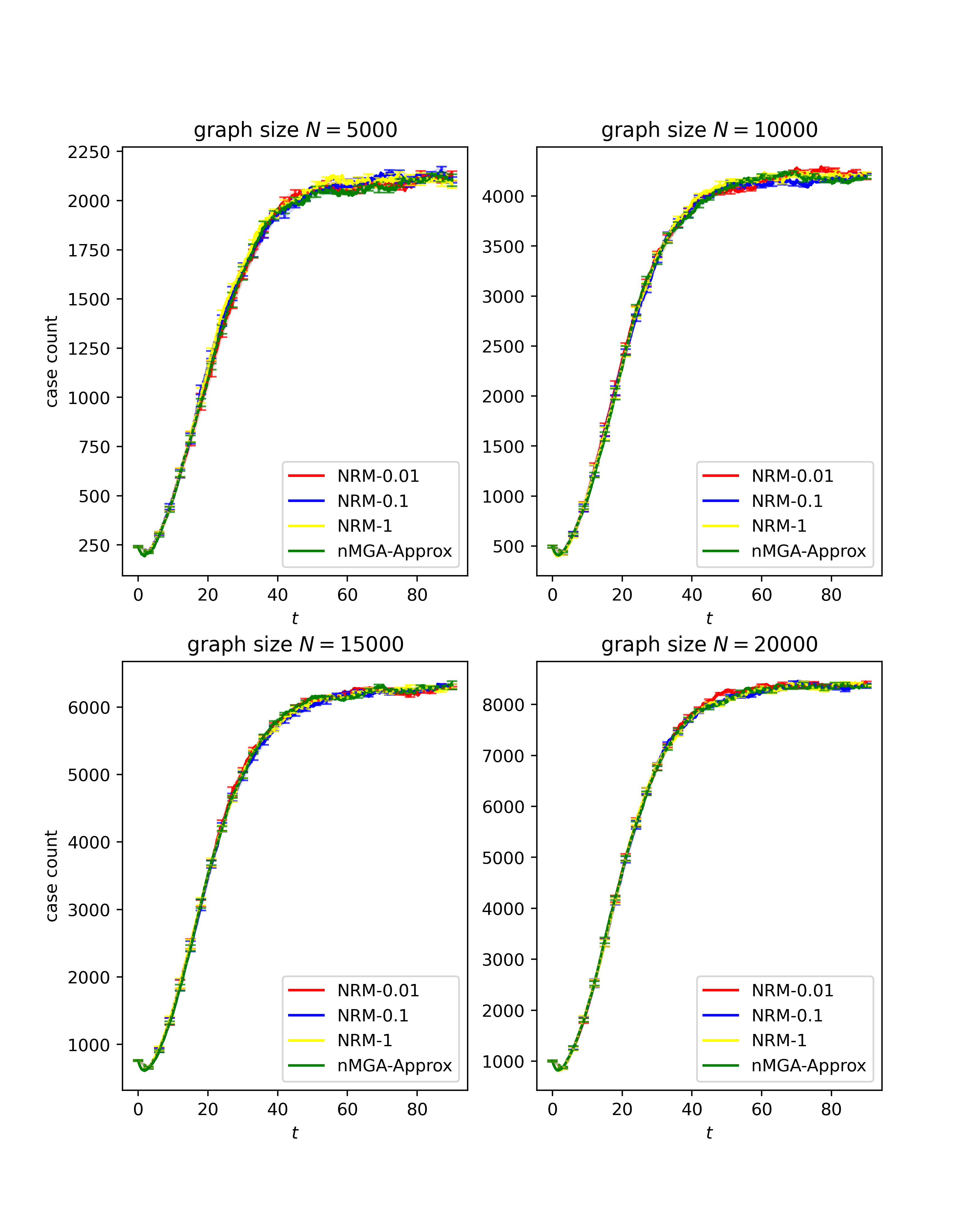} 
  \caption{{\bf Trajectories on graphs with different sizes.} Each trajectory is computed as the average of 5 repeated runs. The infection interevent distribution is configured to be Weibull($\kappa=2, \lambda=10)$. The recovery distribution is Exp$(0.25)$. The cooperation parameter $\sigma$ is $0.5$. When an algorithm is termed ``NRM-$\dt$'', it means it runs NRM with step size $\dt$.}
  \label{figCooperativeTraj}
\end{figure}

The algorithms in Fig.~\ref{figCooperativeTraj} exhibit indistinguishable trajectories. The performance, however, differs greatly across algorithms. It takes nMGA-Approx around 5 minutes to run on graph size 10000 and around 25 minutes to run on graph size 20000, while NRM-0.01 costs 4.3 seconds on graph size 10000 and 8.8 seconds on graph size 20000. As $\dt$ gets larger, the time consumption gets lower, as NRM-0.1 and NRM-1 cost far less than 1 second on all graph sizes considered.

The reasoning behind the performance difference is that with a leading error term of $\dt^2$, NRM achieves the same level of accuracy with fewer number of rate function evaluations. Notably, the midpoint rule is just one of the many quadrature rules that can be applied here. It is possible to utilize other quadrature rules with smaller leading error term, as long as the interpolation step has the same leading error term as the numerical integration. 

\section{Discussion}
We have reviewed simulation algorithms for non-Markovian and Markovian epidemics. We have learned that
\begin{itemize}
    \item NRM~(\ref{alg:nrm}) is the best choice when the simulation task is non-Markovian or Markovian with heterogeneous rates.
    \item Optimized GA~(\ref{alg:OptGA}) is the best choice when the simulation task is Markovian with homogeneous rates.
    \item Through properly designed event emitters, both NRM and optimized GA can be adapted for more involved scenarios, such as time-varying networks and cooperative infections. 
    \item It is possible for node-centric GA~(\ref{alg:nodecentric}) to out-perform other algorithms when the graph size is small (as shown in Fig.~\ref{figPerfCommute}), in spite of its sub-optimal asymptotic behavior (Fig.~\ref{figPerfCommute}). 
\end{itemize}

Simulations with event emitters are not without limitations. Firstly, one must clearly identify the independent components in the simulation task. For instance, in the scenario of cooperative infection, infections from different infected neighbors to the same susceptible node are not independent events when $\sigma\neq 1$, which is why the edge-wise infection emitter $\Inf(x, y)$ cannot be used anymore. In the mean time, infections of different susceptible nodes are independent events, which is why we have $\Inf(x)$ instead. Nevertheless, we believe that this forces practitioners to be explicit about their assumptions on independence, even though it does complicate the design process in the early stage. Secondly, it is not always possible to make the number of created/removed event emitters $\bigO(1)$. For instance, consider simulating an SIS epidemic on a complete graph. If the number of created/removed event emitters is $\bigO(N)$, then the per-iteration cost is at best $\bigO(N)$. 

The technique of grouping event emitters also presents a trade-off between space and time complexity. Suppose the graph has $V$ vertices and $E$ edges. For node-centric GA, the space complexity is $\bigO(V)$ since rates are computed and stored for each node. For optimized GA, the space complexity is $\bigO(V+E)$ since we have to maintain a hashset of all active edges. 

\section{Supporting information}

\paragraph*{S1 Appendix.}
\label{S1_Appendix}
The bookkeeping of elapsed times is, in fact, a design choice for non-Markovian simulations. To see why, consider the scenario in Fig.~\ref{figElapsed}.
\begin{figure}[H]
  \centering
  \includegraphics[scale=0.23]{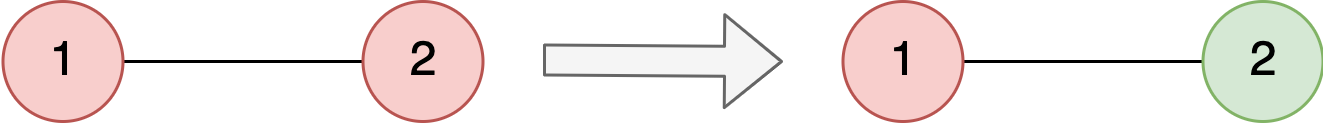}
  \caption{{\bf Two different rules for elapsed time bookkeeping.}}
  \label{figElapsed}
\end{figure}
With the recovery of 2, the event emitter $\Inf(1\to 2)$ needs to be introduced. However, the time to infection may adopt either of the following rules:
\[\begin{cases}
  T_\text{inf} \sim \psi(\cdot\cond 0)  \;\;&\text{rule 1}\\
  T_\text{inf} \sim \psi(\cdot\cond t_1)  \;\;&\text{rule 2,}
\end{cases}\]
where $t_1$ is the duration in which node 1 has been ill. Rule 1 is easier to implement while rule 2 is closer to reality, especially if we accept that the root cause of non-Markovian infection is the development of symptoms through time, which results in varying infectivity. Fortunately, rule 2 only requires slight modifications to Table~\ref{tableSIS}.
\begin{table}[H]
  \centering
  \caption{\textbf{Inf($x\to y$) for an SIS epidemic on a network, with rule 2.}}
  \begin{tabular}{|c|c|}
  \hline
             & Infection                          \\ \hline
  ID         & Inf($x\to y$)                      \\ \hline
  f          & \makecell{$y$.state $\gets$ I \\ $y$.infected\_since $\gets t_\text{exe}$ }      \\ \hline
  C          & \makecell{$\{\text{Rec}(y)\}$\\$\cup$ \\ $\{\text{Inf}(y\to z) \;|\; z\in \N(y), z.\text{state=S}\}$} \\ \hline
  R          & \makecell{$\{\text{Inf}(z\to y)\;|\; z\in \N(y), z\text{.state=I}\}$} \\ \hline
  $\Psi$     & $\Psi(\cdot\cond t_\text{intro} - x.\text{infected\_since})$    \\ \hline
  \end{tabular}
  \begin{flushleft} 
    $t_\text{exe}$: time when $\Inf(x\to y)$ executes (when $y$ gets infected by $x$).\\ 
    $t_\text{intro}$: time when $\Inf(x\to y)$ is created by some other event emitter.
  \end{flushleft}
  \label{tableSISMod}
\end{table}

This extra bookkeeping guarantees that every infected node ``knows'' when it was infected. In practice, this brings minimal changes to the code and can be toggled on/off easily. Thus, rule 2 is adopted for all non-Markovian simulations in this article.

\paragraph*{S2 Appendix.}
\label{S2_Appendix}
\begin{proof}[Proof of Theorem \ref{thm:indep}]
  We prove the theorem by induction on $l$.\\
    \textbf{Base Case} $l=1$\\
    When $l=1$, all tentative times are freshly generated and the theorem holds as a direct consequence of Assumption 1. \\
    \textbf{Inductive Case} $l\geq 2$\\
    Assume for the sake of argument that the theorem holds for all $l'<l$, i.e., $\forall l'<l,$ $\{T^{l'}_i\}_{i\in E^{l'}}$ consists of mutually independent random variables given $\traj{l'-1}$. Particularly, $\{T^{l-1}_i\}_{i\in E^{l-1}}$ are mutually independent given $\traj{l-2}$, thanks to the inductive hypothesis. \\
    Given $\traj{l-1}$, we determine $E^l, \R^l$. If $\R^l=\emptyset$, all tentative times in iteration $l$ are freshly generated and the theorem holds as a direct consequence of Assumption 1. We therefore only consider the non-trivial case where $\R^l\neq\emptyset$.\\
    Let $\J$ be an arbitrary subset of $E^l$, allowing us to define $\J^l\coloneqq\R^l\cap\J, \bar{\J}^l\coloneqq\J\setminus\J^l$. Because of Assumption~\ref{assumption:assumption}, we factorize the joint CDF as
    \begin{align}
      \Pr\{\vec{T}^l_\J\leq \vec{t}_\J\}
    = \left(\prod_{i\in\bar\J^l}\Pr\{T^l_i\leq t_i\cond\traj{l-1}\}\right) \cdot \Pr\{\vec{T}^l_{\J^l}\leq \vec{t}_{\J^l}\cond \traj{l-1}\}. \label{eqn:splitFR}
    \end{align}
    If $\J^l=\emptyset$, i.e. if all tentative times in $\J$ are freshly generated in iteration $l$, the second factor in Eq~(\ref{eqn:splitFR}) disappears and the factorization is already complete. We therefore only consider the non-trivial case where $\J^l\neq\emptyset$ and proceed to factorize the second factor.
    \begin{align}
      &\Pr\{\vec{T}^l_{\J^l}\leq \vec{t}_{\J^l}\cond \traj{l-1}\} = \Pr\{\vec{T}^{l-1}_{\J^l}\leq \vec{t}_{\J^l}\cond \traj{l-1}\} \;\;\;\;\;\; (\J^l \text{ is reused}) \label{eqn:secondTerm}\\
    =& \Pr\{\vec{T}^{l-1}_{\J^l}\leq \vec{t}_{\J^l}\cond \traj{l-2}, K^{l-1}=k^{l-1}, T^{l-1}=t^{l-1}\} \;\;\;\;\;\; (\text{Defn.}~\ref{def:traj}) \nonumber\\
    =& \lim_{\epsilon\to 0}\frac{\Pr\{\vec{T}^{l-1}_{\J^l}\leq \vec{t}_{\J^l}, K^{l-1}=k^{l-1}, \exenbr \cond \traj{l-2}\}}{\Pr\{K^{l-1}=k^{l-1}, \exenbr \cond \traj{l-2}\}}\label{condprobdef}\\
    =& \lim_{\epsilon\to 0}\frac{\Pr\{\vec{T}^{l-1}_{\J^l}\leq \vec{t}_{\J^l};\; \vec{T}^{l-1}_{E^{l-1}\setminus\{k^{l-1}\}} > t^{l-1};\; \tentnbr \cond \traj{l-2}\}}{\Pr\{\vec{T}^{l-1}_{E^{l-1}\setminus\{k^{l-1}\}} > t^{l-1};\; \tentnbr \cond \traj{l-2}\}}\label{rewritingKT}\\
    =& \lim_{\epsilon\to 0}\frac{\Pr\{t^{l-1} < \vec{T}^{l-1}_{\J^l} \leq \vec{t}_{\J^l};\; \vec{T}^{l-1}_{E^{l-1}\setminus(\J^l\cup\{k^{l-1})\}} > t^{l-1};\; \tentnbr \cond\traj{l-2}\}} {\Pr\{\vec{T}^{l-1}_{\J^l} > t^{l-1};\; \vec{T}^{l-1}_{E^{l-1}\setminus(\J^l\cup\{k^{l-1}\})} > t^{l-1};\; \tentnbr \cond\traj{l-2}\}} \label{inductMany}\\
    =&  \lim_{\epsilon\to 0}\frac{\Pr\{t^{l-1} < \vec{T}^{l-1}_{\J^l} \leq \vec{t}_{\J^l};\; \vec{T}^{l-1}_{E^{l-1}\setminus(\J^l\cup\{k^{l-1}\})} > t^{l-1};\; \tentnbr \cond\traj{l-2}\}} {\Pr\{\vec{T}^{l-1}_{\J^l} > t^{l-1};\; \vec{T}^{l-1}_{E^{l-1}\setminus(\J^l\cup\{k^{l-1}\})} > t^{l-1};\; \tentnbr \cond\traj{l-2}\}} \nonumber\\
    =& \frac{\Pr\{t^{l-1} < \vec{T}^{l-1}_{\J^l} \leq \vec{t}_{\J^l}\cond \traj{l-2}\}}{\Pr\{\vec{T}^{l-1}_{\J^l} > t^{l-1}\cond \traj{l-2}\}} \label{eqn:useIH1}\\
    =& \prod_{j\in\J^l} \frac{\Pr\{t^{l-1} < T^{l-1}_j \leq t_j\cond \traj{l-2}\}}{\Pr\{T^{l-1}_j > t^{l-1}\cond \traj{l-2}\}} \label{eqn:useIH2}\\
    =& \prod_{j\in\J^l} \Pr\{T^{l-1}_j\leq t_j\cond T^{l-1}_j > t^{l-1}, \traj{l-2}\},\label{eqn:manyReused}
    \end{align}
    where Eq~(\ref{condprobdef}) is essentially the definition of conditional probability, Eq~(\ref{rewritingKT}) rewrites the event $K^{l-1}=k^{l-1}, \exenbr$, and Eq~(\ref{inductMany}) simply regroups indices. $\lim_{\epsilon\to 0}$ is introduced in Eq~(\ref{condprobdef}) to prevent ``0/0''. By assuming all $T^l_i$ to be continuous random variables, we know this limit exists. Eq~(\ref{eqn:useIH1}) and Eq~(\ref{eqn:useIH2}) make use of the inductive hypothesis. \\
    Since we have assumed that $\R^l\neq\emptyset$, we take $j\in\R^l$ and pick $\J=\J^l=\{j\}$. Since the choice of $\J\subset E^l$ is arbitrary, the reasoning from Eq~(\ref{eqn:secondTerm}) to Eq~(\ref{eqn:manyReused}) still applies, giving us
    \begin{align}
      \Pr\{T^l_j\leq t\cond \traj{l-1}\} = \Pr\{T^{l-1}_j\leq t\cond T^{l-1}_j > t^{l-1}, \traj{l-2}\}.\label{eqn:oneReused}
    \end{align}
    Substitute Eq~(\ref{eqn:oneReused}) back into Eq~(\ref{eqn:manyReused}), yielding
    \begin{align}
      \Pr\{\vec{T}^l_{\J^l}\leq \vec{t}_{\J^l}\cond \traj{l-1}\} = \prod_{j\in\J^l} \Pr\{T^l_j\leq t_j\cond \traj{l-1}\}.\label{eqn:reuseFactorized}
    \end{align}
    Plug Eq~(\ref{eqn:reuseFactorized}) into Eq~(\ref{eqn:splitFR}) and the factorization is complete, which indicates mutual independence, concluding the proof. 
\end{proof}

\paragraph*{S3 Appendix.}
\label{S3_Appendix}
\begin{proof}[Proof of Lemma~\ref{lemma:TDist}]
  Again, we prove Lemma~\ref{lemma:TDist} by induction on $l$. 
  % Assume for the sake of argument that $\forall l<L$, Eq~(\ref{eqn:TDistLong}) holds. We now discuss the distribution of $T^L_i$
  \begin{itemize}
    \item \textbf{Base Case} $l=1$\\
    $T^1_i$ can only be freshly generated, meaning that $t^1_i=t^0=0$. Eq~(\ref{eqn:TDistLong}) follows from 
    \[\Pr\{T^1_i\leq t\cond\traj{l-1}\} = \Psi_i(t).\] 
    \item \textbf{Inductive Case} $l\geq 2$\\
    Assume that $\forall l'<l$, $\Pr\{T^{l'}_i\leq t\cond\traj{l-1}\}$ takes the form of Eq~(\ref{eqn:TDistLong}).
    \begin{itemize}
      \item If $T^l_i$ is freshly generated, then $t^l_i = t^{l-1}$, meaning that $\Psi_i(t^{l-1}-t^l_i) = 0$. As a result,
      \begin{align*}
        \Pr\{T^l_i\leq t\cond\traj{l-1}\} 
        = \Psi_i(t-t^{l-1}) = \frac{\Psi_i(t-t^l_i) - \Psi_i(t^{l-1}-t^l_i)}{1 - \Psi_i(t^{l-1} - t^l_i)}.
      \end{align*}
      \item Otherwise, if $T^l_i$ is not freshly generated, meaning that event emitter $i$ already exists in iteration $l-1$ (hence $t^{l-1}_i=t^l_i$), and we can let $s\coloneqq t^{l-1}_i=t^l_i$.
      Because of the inductive hypothesis, $T^{l-1}_i$ satisfies
      \[\Pr\{T^{l-1}_i\leq t\cond\traj{l-2}\} = \frac{\Psi_i(t-s) - \Psi_i(t^{l-2}-s)}{1 - \Psi_i(t^{l-2} - s)}.\]
      Thanks to Corollary~\ref{corollary:inductOne},
      \begin{align*}
        \Pr\{T^l_i\leq t\cond\traj{l-1}\} 
        =& \Pr\{T^{l-1}_i\leq t\cond T^{l-1}_i > t^{l-1}, \traj{l-2}\}\\
        =& \frac{\Pr\{T^{l-1}_i \leq t\cond\traj{l-2}\} - \Pr\{T^{l-1}_i\leq t^{l-1}\cond\traj{l-2}\}}{1 - \Pr\{T^{l-1}_i \leq t^{l-1}\cond\traj{l-2}\}}\\
        =& \frac{\Psi_i(t-s)-\Psi_i(t^{l-1}-s)}{1-\Psi_i(t^{l-1}-s)},
      \end{align*} 
      which is exactly what we need.
    \end{itemize} 
  \end{itemize}
\end{proof}

\paragraph*{S4 Appendix.}
\label{S4_Appendix}
\begin{proof}[Proof of Theorem \ref{thm:equiv}]
  Given the previous trajectory up to $l-1$, consider the tentative times $\{T^l_i\}_{i\in E^l}$ of NRM. We have that
  \begin{itemize}
    \item $\{T^l_i\}_{i\in E^l}$ are mutually independent given $\traj{l-1}$ (Theorem \ref{thm:indep}).
    \item $\Pr\{T^l_i\leq t\cond\traj{l-1}\} = \Psi_i(t-t^{l-1}\cond t^{l-1}-t^l_i)$ (Eq~(\ref{eqn:tdist})).
  \end{itemize}
  For NRM, the time of the $l$-th execution is $T^l\coloneqq\min_i T^l_i$. We obtain
  \begin{align*}
    \Pr\{T^l > t\cond\traj{l-1}\}
  =&\Pr\{\min_i T^l_i > t\cond\traj{l-1}\} = \Pr\{T^l_{E^l} > t\cond\traj{l-1}\}\\
  =& \prod_i \Pr\{T^l_i > t\cond\traj{l-1}\} 
  = \prod_i \left[1 - \Psi_i(t-t^{l-1}\cond t^{l-1} - t^l_i)\right].
  \end{align*}
  Meanwhile for nMGA, the time of the $l$-th execution $\tilde T^l\coloneqq t^{l-1} + \Delta^l$ satisfies
  \begin{align*}
    \Pr\{\tilde T^l > t\cond\traj{l-1}\} 
  =& \Pr\{\Delta^l + t^{l-1} > t\cond\traj{l-1}\} = \Pr\{\Delta^l > t-t^{l-1}\cond\traj{l-1}\} \\
  =& \Phi^l(t-t^{l-1}) = \prod_i \left[1-\Psi_i(t-t^{l-1}\cond t^{l-1}-t^l_i)\right] \\
  =& \Pr\{T^l > t\cond\traj{l-1}\},
  \end{align*}
  and the fact that $\forall t, \Pr\{T^l > t\cond\traj{l-1}\}=\Pr\{\tilde T^l > t\cond\traj{l-1}\}$ implies $f_{T^l} \equiv \tilde f_{T^l}$. We see that the times of execution, $T^l$ and $\tilde T^l$, do have the same (marginal) distribution. We then investigate the distribution of $K^l$ given trajectory $\traj{l-1}$ and $T^l=t^l$ in NRM. 
  \begin{align*}
    \Pi^l(k^l\cond t^l)
  \coloneqq& \Pr\{K^l = k^l\cond T^l =t^l, \traj{l-1}\}\\
  =& \lim_{\epsilon\to 0}\frac{\Pr\{K^l=k^l, t^l \leq T^l < t^l + \epsilon\cond\traj{l-1}\}}{\Pr\{t^l \leq T^l < t^l + \epsilon\cond\traj{l-1}\}}\\
  =& \lim_{\epsilon\to 0}\frac{\Pr\{K^l=k^l, t^l \leq T^l < t^l + \epsilon\cond\traj{l-1}\}}{\sum_{k} \Pr\{K^l=k, t^l \leq T^l < t^l + \epsilon\cond\traj{l-1}\}},
  \end{align*}
  where each individual $\Pr\{K^l=k, t^l \leq T^l < t^l + \epsilon\cond\traj{l-1}\}$ can be expanded into 
  \begin{align}
    &\Pr\{K^l=k, t^l \leq T^l < t^l + \epsilon\cond\traj{l-1}\} = \Pr\{t^l \leq T^l_k < t^l + \epsilon;\; \forall j\neq k, T^l_j > t^l\cond\traj{l-1}\} \nonumber\\
  =& \psi_k(t^l - t^{l-1}\cond t^{l-1} - t^l_k) \cdot \epsilon \cdot \prod_{j\neq k} \left[1 - \Psi_j(t^l-t^{l-1}\cond t^{l-1} - t^l_j)\right]  \quad (\text{Thm.}~\ref{thm:indep}, \text{Lemma}~\ref{lemma:TDist}) \nonumber\\
  =& \frac{\psi_k(t^l - t^{l-1}\cond t^{l-1} - t^l_k)}{1-\Psi_k(t^l-t^{l-1}\cond t^{l-1}-t^l_k)} \cdot \epsilon \cdot \prod_j \left[1 - \Psi_j(t^l-t^{l-1}\cond t^{l-1} - t^l_j)\right] \nonumber\\
  =& \frac{\psi_k(t^l - t^l_k)}{1 - \Psi_k(t^l-t^l_k)} \cdot \epsilon \cdot \prod_j \left[1 - \Psi_j(t^l-t^{l-1}\cond t^{l-1} - t^l_j)\right] \nonumber\\
  =& \psi_k(0\cond t^l-t^l_k) \cdot \epsilon \cdot \underbrace{\prod_j \left[1 - \Psi_j(t^l-t^{l-1}\cond t^{l-1} - t^l_j)\right]}_{\text{not dependent on } k}\label{eqn:notdeponk}.
  \end{align}
  Since both $\epsilon$ and the factor not dependent on $k$ in Eq~(\ref{eqn:notdeponk}) cancel out, for NRM, we have
  \begin{align*}
    \Pi^l(k^l) = \frac{\psi_{k^l}(0\cond t^l-t^l_{k^l})}{\sum_k\psi_k(0\cond t^l-t^l_k)},
  \end{align*}
  which is the same as the PMF weights stipulated by nMGA in Alg.~\ref{alg:nMGAExact} and Eq~(\ref{eqn:nmgaPMF}). 

  In conclusion, 
  \begin{align*}
    \begin{rcases}
      f_{T^l}\equiv \tilde f_{T^l}&\\
      \Pi^l\equiv \tilde\Pi^l&
    \end{rcases} \implies f^{\text{NRM}}_{T^l,K^l}(\cdot\cond\traj{l-1}) \equiv f^{\text{nMGA}}_{T^l,K^l}(\cdot\cond\traj{l-1}),
  \end{align*}
  which indicates equivalence, concluding the proof. 
\end{proof}

% \begin{figure}
% 	\centering
% 	\fbox{\rule[-.5cm]{4cm}{4cm} \rule[-.5cm]{4cm}{0cm}}
% 	\caption{Sample figure caption.}
% 	\label{fig:fig1}
% \end{figure}

% \subsection{Tables}
% See awesome Table~\ref{tab:table}.

% \begin{table}
% 	\caption{Sample table title}
% 	\centering
% 	\begin{tabular}{lll}
% 		\toprule
% 		\multicolumn{2}{c}{Part}                   \\
% 		\cmidrule(r){1-2}
% 		Name     & Description     & Size ($\mu$m) \\
% 		\midrule
% 		Dendrite & Input terminal  & $\sim$100     \\
% 		Axon     & Output terminal & $\sim$10      \\
% 		Soma     & Cell body       & up to $10^6$  \\
% 		\bottomrule
% 	\end{tabular}
% 	\label{tab:table}
% \end{table}

\bibliographystyle{unsrtnat}
\bibliography{references}  %%% Uncomment this line and comment out the ``thebibliography'' section below to use the external .bib file (using bibtex) .





% \end{thebibliography}

\end{document}